\documentclass[a4paper,11pt,leqno]{amsart}

\usepackage[includeheadfoot,a4paper,left=2cm,right=2cm,top=2cm,bottom=2cm]{geometry}
\usepackage{mathrsfs}
\usepackage{url}
\usepackage{enumerate}

\makeatletter
\usepackage{hyperref}
\usepackage{amsfonts,dsfont,delarray,amssymb,amsmath,amsthm,color}
\usepackage[applemac]{inputenc}
\usepackage{esint}
\usepackage{stmaryrd}
\newtheorem{theorem}{Theorem}
\newtheorem{proposition}{Proposition}[section]
\newtheorem{lemma}[proposition]{Lemma}
\newtheorem{corollary}[theorem]{Corollary}
\newtheorem{remark}[proposition]{Remark}
\newtheorem{definition}[proposition]{Definition}

\newcommand{\op}[1]{{\rm{#1}}}
\newcommand{\ep}{\epsilon}
\newcommand{\R}{\mathbb R}
\newcommand{\N}{\mathbb N}
\newcommand{\yg}{|y|^\gamma}
\newcommand{\w}{\mathcal W_\alpha}
\newcommand{\nab}{\nabla}

\newcommand{\loc}{{\text{loc}}}
\newcommand{\ve}{\varepsilon}

\DeclareMathOperator*{\argmin}{arg\,min}

\title{Equidistribution of jellium energy for Coulomb and Riesz Interactions}

\author[M. {Petrache}]{Mircea {Petrache}$^1$}
\address{$^1$ Max-Planck Institute for Mathematics,
Vivatstrasse 17, 53111, Bonn, Germany.}
\email{mircea.petrache@upmc.fr}

\author[S. {Rota Nodari}]{Simona {Rota Nodari}$^2$}
\address{$^2$ Institut de Math\'ematiques de Bourgogne (IMB), CNRS, UMR 5584,
Universit\'e Bourgogne Franche-Comt\'e,
F-21000 Dijon, France.}
\email{simona.rota-nodari@u-bourgogne.fr}
\thanks{The authors are grateful to  S. Serfaty and to S. Torquato for helpful discussions and useful comments.}

\begin{document}
\maketitle
\begin{abstract}
\noindent For general dimension $d$ we prove the equidistribution of energy at the micro-scale in $\mathbb R^d$, for the optimal point configurations appearing in Coulomb gases at zero temperature. At the microscopic scale, \emph{i.e.} after blow-up at the scale corresponding to the interparticle distance, in the case of Coulomb gases we show that the energy concentration is precisely determined by the macroscopic density of points, independently of the scale. This uses the ``jellium energy'' which was previously shown to control the next-order term in the large particle number asymptotics of the minimum energy. As a corollary, we obtain 
sharp error bounds on the discrepancy between the number of points and its expected average of optimal point configurations for Coulomb gases, extending previous results valid only for $2$-dimensional log-gases. For Riesz gases with interaction potentials $g(x)=|x|^{-s}, s\in]\min\{0,d-2\},d[$ and one-dimensional log-gases, we prove the same equidistribution result under an extra hypothesis on the decay of the localized energy, which we conjecture to hold for minimizing configurations. In this case we use the Caffarelli-Silvestre description of the non-local fractional Laplacians in $\mathbb R^d$ to localize the problem.
\end{abstract}

\section{Introduction}
A long-standing question and direction of research at the intersection of mathematics and physics is to ask how solving the minimization problem of sums of two-body interactions between a large number of particles, or more simply between a large number of points, can lead to ``collective behavior'' of the minimizers, in which some better order structure is seen to emerge. A type of emergent phenomenon, in which a more rigid structure for minimizers tends to diminish the overall complexity of the configurations and is observed empirically in a large number of situations, is usually termed ``crystallization''. This name refers in the most restrictive meaning to the appearance of periodic structures for minimizers (see the recent review \cite{blanclewin}). From a statistical physics viewpoint, the case in which we have a more ordered structure with higher correlations than a random one fits within the framework of crystallization.\\

The particular model which we consider here comes from the theory of Riesz gases already studied in \cite{gl13,ss2d,ss1d,rs,ps}. At zero temperature, we individuate and rigorously prove a \emph{rigidity phenomenon} which is a weak version of crystallization. A corollary of our results is that minimizers of the renormalized Coulomb gases are very uniform configurations (see Theorem~\ref{discrepancybound} below). The present work extends the result \cite{rns} valid for the $2$-dimensional Coulomb gases to the case of general dimension $d$ and of Riesz gases with power-law interactions with power $s\in[\min\{0,d-2\}, d[$, using the strategies for localizing the energy available in \cite{ps}, and inspired by \cite{cafsil}. In particular, the present result completes the parallel between the work of Rota Nodari and Serfaty \cite{rns} and the one of Alberti, Choksi and Otto \cite{aco} to general dimensions, for the case $s=d-2$. A consequence of this parallel and of the result of the present work, is the conjecture that \cite{aco} might have extensions to nonlocal interactions corresponding to Green functions for the fractional laplacian. In two dimensions the ``Abrikosov conjecture'' \cite{gl13, ps} valid in this range of exponents is that the renormalized energy is in fact minimized by a suitably rescaled copy of the triangular lattice $\mathbb Z + e^{i\pi/3}\mathbb Z$. An analogous conjecture holds for the minimizers of the Ohta-Kawasaki type model of \cite{aco}. It is believed that in high enough dimension the lattice structure is not characteristic of minimizing configurations.\\

Crystallization problems have up to now been solved only for specific short range interaction potentials (see \cite{The, BPT, HR, Sut, Rad} and references therein) that do not cover Coulomb forces, or in 1D systems \cite{BL, Kun, ss1d}. As a positive result, in \cite{gl13,ps} it was shown however that in dimension 2 and for the above range of exponents $s$, if the minimizer is a lattice, then it has to be the triangular one.\\

Recently, the study of one-component plasmas at positive temperature has received a lot of attention. Results to some extent parallel to ours in the case of $2$-dimensional Coulomb gases was provided by Bauerschmidt-Bourgade-Nikula-Yau \cite{bourgade} and Lebl\'e-Serfaty \cite{lebleserf2}, whose main result is interpretable as a quantification of discrepancy at the microscopic scale at positive temperature. Again in the $2$-dimensional case at positive temperature, bootstrapping techniques similar to ours have allowed Lebl\'e \cite{leb2d} to prove microscopic energy distribution results. Other results in the same spirit concern universality of the law of eigenvalues of random matrices \cite{ben1998large,bourgade2014edge,bourgade2012bulk,bourgade2014universality,bourgade2014local,valko2009continuum}. In the random setting the analogue of our discrepancy bounds (see Theorem~\ref{discrepancybound} below) are so-called charge fluctuation results, see also \cite{zabrodin2006large} and the references therein.

\subsection{General setting of the problem}
We now pass to the precise description of our problem. We study the equilibrium properties of  a system of $n$  points in the full space of dimension $d\ge 1$, interacting via Riesz kernel interactions and confined by an ``external field" or potential $V$. More precisely, we consider energies of the form  
\begin{equation} \label{Hn}
H_n(x_1, \cdots, x_n)=   \sum_{i\neq j} g (x_i-x_j)  +n \sum_{i=1}^n V(x_i)
\end{equation}
where  $x_1 , \cdots, x_n$ are $n$ points in $\mathbb R^d$ and the interaction kernel is given by either
\begin{equation}\label{kernel}
 g(x)=\frac{1}{|x|^{s}}\qquad \max(0, d-2)\le s<d, 
 \end{equation} 
or  
\begin{equation}\label{wlog}
g(x)=-\log |x| \quad \text{in dimension } d=1,
\end{equation} 
or 
\begin{equation}\label{wlog2d}
g(x) = - \log |x| \quad \text{in dimension } d=2.
\end{equation}
In the mean-field setting, the factor $n$ multiplying the one-body potential term has the role of giving equal influence to this term as compared to the two-body interaction term. If $V$ has some particular homogeneity, then often we can reduce to an energy of this form by an appropriate scaling. The case of $s\in [d-2,d[, s<0$, which is not treated here, can happen only for $d=1$ and this seems to be a more tractable situation. In particular the case $s=-1$, \emph{i.e.} $g(x)=|x|$, was shown to be completely solvable \cite{aizenmanmartin, Len1, Len2, BL, Kun}. Note that, in what follows, we will take the convention that $s=0$ when $g(x)=-\log|x|$, \emph{i.e.} in the cases \eqref{wlog} and \eqref{wlog2d}.
\par The reason why systems of particles with Coulomb and Riesz interactions are interesting in \emph{statistical physics} is that they represent the most basic model containing the long-range interaction potentials typical of the electrostatic potential. For studies in the Coulomb case see \cite{SM, LO, JLM, psmith}, and  see \cite{Ser} for a review. The possibility of changing the exponent $s$ allows to ``turn on'' or ``off'' the locality of the interactions. The case $s\ge d$ (also called \emph{hypersingular} case \cite{sk, bhsreview}) corresponds to interactions of more local nature. In \cite{gl13}, the precise energy of our form is linked to the study of vortex systems, that appear in classical and quantum fluids \cite{clmp, cy, cry} and in fractional quantum Hall physics \cite{Gir, RSY1, RSY2}.
\par Our interaction energy is also appearing in the theory of \emph{random matrix ensembles} \cite{agz, meh, dei}. Worth mentioning, especially the Ginibre ensembles, as was described in \cite{ginibre1965statistical, Wi2} and exploited by Dyson starting in \cite{dyson1} for $d=2,\ s=0$ and GOE or GUE for $d=1,\ s=0$, as described in \cite{forrester2010log} and the references therein. 
\par Another direction of study in which this type of energy appears is related to Smale's $7$th problem \cite{smalenextcentury}, which asks to find fast algorithms for minimizing our energy up to a very small error. Studies of this question are related to the optimal conditioning for interpolation points and to the theory of quadrature (see \cite{shub, sk, safftotik} and the references therein).\\

The leading order behavior of minimizers of $H_n$ is known: there holds
\begin{equation*}
\frac{1}{n} \sum_{i=1}^n\delta_{x_i}\rightharpoonup \mu_V,
\end{equation*}
where the convergence is the weak convergence of probability measures, and $\mu_V$ is the equilibrium measure, \emph{i.e.} the minimizer of the energy
\begin{equation}\label{defI}
I(\mu):=\int\int_{\mathbb R^d\times\mathbb R^d}g(x-y)d\mu(x)d\mu(y) + \int_{\mathbb R^d}V(x)d\mu(x)\ .
\end{equation}
The next-order behavior of $H_n$ and of its minimizers is observed at the scale $n^{-1/d}$ at which (after blow-up) the points become well-separated. As first observed in \cite{ss2d}, \cite{ss1d} via methods later extended in \cite{rs} and \cite{ps} to our general setting, if $\mu_V$ is the minimizer of $I$, then $H_n$ can be split into two contributions corresponding to a constant leading order term and a typically next order term as follows:
\begin{equation}\label{split}
H_n(x_1, \dots, x_n)= n^2 \mathcal I (\mu_V) +2 n \sum_{i=1}^n \zeta(x_i)+  n^{1+s/d}w_n(x_1, \dots, x_n)
\end{equation} in the case \eqref{kernel}
and respectively 
\begin{equation}\label{splitlog}
H_n(x_1,\dots, x_n)= n^2 \mathcal I (\mu_V) - \frac{n}{d} \log n  + 2n\sum_{i=1}^n \zeta(x_i) + nw_n(x_1, \dots, x_n)
 \end{equation}
in the cases \eqref{wlog} or \eqref{wlog2d}, where $w_n$ will be made explicit in Proposition \ref{splitting} and $\zeta $ is an ``effective potential'' (defined below in \eqref{defzeta}) depending only on $V$, which is nonnegative and vanishes on $\op{supp}(\mu_V):=\Sigma$. As shown in \cite{ss2d,ss1d,rs,ps}, $w_n$ has a limit $\mathcal W$ as $n\to \infty$, which is our renormalized energy. The precise definition of $\mathcal W$ is given in Section~\ref{secrenoren} below in terms of the potential generated by the limits of configurations blown-up at the scale $n^{1/d}$. The renormalization procedure consists in considering first a version of the energy where the charges are ``spread'' at scale $\eta$, and defining an energy $\mathcal W_\eta$ as the version of $\mathcal W$ where the self-interaction term of these spread charges, which becomes infinite as $\eta\to 0$, is removed. Then $\mathcal W$ is defined as $\lim_{\eta\to 0} \mathcal W_\eta$. Due to the splitting formulas \eqref{split}, \eqref{splitlog}, minimizers of $H_n$ converge, after blow-up, to minimizers of $\mathcal W$. As mentioned above, it is a hard mathematical conjecture corroborated by simulations and experimental evidence (the so-called ``Abrikosov conjecture'' in $2$-dimensions being the most celebrated case), that in low dimensions the minimum of $\mathcal W$ is achieved at simple crystalline configurations, \emph{i.e.} minimizers of $\mathcal W$ are expected to ressemble perfect lattices. In \cite{ss2d,ss1d,rs,ps} the analysis of the microscopic behavior of minimizers of $H_n$ was thus connected to the behavior of minimizers of $\mathcal W$ by allowing to rigorously formulate the crystallization conjecture in terms of $\mathcal W$.

\subsection{Statement of the main results}
We now state the main results of our paper. As said before, these results are the generalization of the result of \cite{rns} for the $2$-dimensional Coulomb gases to the case of general dimension $d$ and of Riesz gases with power-law interactions. More precisely, we prove that the renormalized energy $\mathcal W$ is equidistributed at the microscopic scale in an arbitrary square provided that the square is chosen sufficiently far away from $\partial \Sigma$. Moreover, we improve the result of \cite[Thm. 4]{ps}, where it was established that almost minimizers of $H_n$ tend to minimize $\mathcal W$ after blow-up at scale $n^{1/d}$ around \emph{almost every} point in $\Sigma$. Here we show that if we deal with a minimizer of $H_n$ this holds after blow-up around \emph{any} point sufficiently far from the boundary of $\Sigma$ (see Section \ref{assumptionsdefinitions} below for precise definitions). Note that for the case $k=1$ we require the extra condition \eqref{hypdecayfields}, which will be discussed in Section \ref{discussextrahyp} below. We conjecture that this hypothesis is automatically verified for sequences of minimizing configurations, but it seems to be out of reach of the present methods. We expect that fundamentally new methods and ideas will be needed for proving this conjecture.

\begin{theorem}\label{mainthm}
Let $(x_1,\cdots,x_n)$ be a minimizer of $H_n$. Let $\mu_V=m_V(x)dx,\ \mu_V'=m_V'(x')dx'$ be respectively the equilibrium measure and its blow-up at scale $n^{1/d}$. Let $\Sigma'$ be the support of $\mu_V'$ and suppose that $m_V\in C^{0,\alpha}(\Sigma)$ for some $\alpha\in ]0,1]$.\\

Let $E_n'=\nabla h_n'$ be the vector fields expressed as the gradient of the potentials of blow-up configurations corresponding to these minimizers, as in \eqref{hnp} below. If $k=1$,  assume that
\begin{equation}\label{hypdecayfields}
\lim_{t\to\infty}\lim_{R\to\infty}\lim_{n\to\infty}\frac{1}{|K_R|}\int_{K_R\times(\R\smallsetminus[-t,t])}\yg|E_n'|^2=0
\end{equation}
uniformly with respect to the choice of the centers of the hypercubes $K_R$.

If $k=0$, there exists $q\in ]0,1[$ such that for $a_n\in \Sigma'$ in the regime where $\op{dist}(K_\ell(a_n),\partial\Sigma')\ge n^{q/d}$, we have

\begin{equation}\label{mainlimit}
\lim_{\eta\to 0} \limsup_{\ell\to\infty} \limsup_{n\to\infty}\left|\frac{\mathcal W_\eta(E_n',K_\ell(a_n))}{|K_\ell|} -\frac{1}{|K_\ell|}\int_{K_\ell(a_n)}\min_{\mathcal A_{m_V'(x')}}\mathcal W dx'\right|=0\,.
\end{equation}

If $k=1$, for every $\varepsilon_0 >0$ there exists a convergence regime depending on \eqref{hypdecayfields} and compatible with the condition  
 $\op{dist}(K_\ell(a_n),\partial\Sigma')\ge \varepsilon_0 n^{1/d}$ for $a_n\in \Sigma'\times\{0\}$ such that \eqref{mainlimit} holds.
\end{theorem}
In the above result it is natural to ask under which conditions we can interchange the renormalization limit $\eta\to0$ with the other ones, obtaining a result valid for $\mathcal W$ rather than for the family $\mathcal W_\eta$. Our proof strategy for the above result is to select ``good boundaries'',  and then use a screening procedure like in \cite{ps}, in order to compare different boundary conditions for the minimizers. In this case the requirement for a ``good boundary'' is that the field $E_\eta$ should not have a large concentration of energy on such boundaries.
\par Unfortunately the purely energetic considerations which we apply in our proof make it impossible to control whether or not the locations of the supports $\partial B(p,\eta),\ p\in\Lambda$ of the smeared charges $\delta_p^{(\eta)}$ appearing in the second term in \eqref{Weta} ``follow'' the energy concentration of $E_\eta$ locally near such good boundaries, and governed by the first term in \eqref{Weta}. In this sense the definition \eqref{Weta} of our energy is really just a global one, and it may happen that large discrepancies between the behaviors $\mathcal W_\eta(K_\ell(a))$ and $\int_{K_\ell(a)\times\mathbb R^k}\yg|E_\eta|^2$ occur for exceptional choices of $K_\ell(a)$. This lack of control prevents the exchange of the $\eta\to0$ limit with the $n,\ell\to\infty$ limit without further assumptions on $K_\ell(a)$.
\par However, if we allow ourselves to slightly perturb the cubes and if we use the charge separation result of Proposition~\ref{separation}, stated below, we can perform the desired interchange of limits for the perturbed hyperrectangles, and we obtain the following:
\begin{theorem}\label{secondmainthm}
Let $(x_1,\cdots,x_n)$ be a minimizer of $H_n$. Let $\mu_V=m_V(x)dx,\ \mu_V'=m_V'(x')dx'$ be respectively the equilibrium measure and its blow-up at scale $n^{1/d}$ as above. Let $\Sigma'$ be the support of $\mu_V'$ and suppose that $m_V\in C^{0,\alpha}(\Sigma)$ for some $\alpha\in ]0,1]$.\\

Let $E_n'=\nabla h_n'$ be a sequence of blown-up vector fields corresponding to these minimizers as in \eqref{hnp} below. If $k=1$,  assume that \eqref{hypdecayfields} holds uniformly with respect to the choice of the centers of the hypercubes $K_R$.

In either of the regimes valid for cases $k=0,\ k=1$ and linking $a_n,\ell,n$ like in Theorem~\ref{mainthm} there exists sets $\Gamma_n$ which can be expressed as bi-Lipschitz deformations $f_n:K_\ell(a_n)\to\Gamma_n$ such that $\|f_n-id\|_{L^\infty}\le 1$ and such that we have
\begin{equation}\label{mainlimitbis}
\limsup_{\ell\to\infty} \limsup_{n\to\infty} \left|\frac{\mathcal W(E_n',\Gamma_n)}{|\Gamma_n|} -\frac{1}{|\Gamma_n|}\int_{\Gamma_n}\min_{\mathcal A_{m_V'(x')}}\mathcal W dx'\right|=0\,.
\end{equation}
Moreover, we may assume that $\Gamma_n$ is a hyperrectangle.
\end{theorem}

Remark that in both theorems the result is slightly weaker in the case $k=1$. As we will explain below, this is due to the fact that we do not know the decay in the extra dimension $y$ of the energy vector fields $E'_n$ corresponding to minimizing configurations of points. This is why we need hypothesis \eqref{hypdecayfields}. If we were able to prove that the l.h.s. of \eqref{hypdecayfields} decays to zero as a negative power of $t$, then \eqref{mainlimit} and \eqref{mainlimitbis} would hold in the regime where $\op{dist}(K_\ell(a_n),\partial\Sigma')\ge n^{q/d}$ for some $q\in]0,1[$.

\subsection{Bound on discrepancy}\label{sechypuniform}
As a consequence of the $k=0$ case of Theorem~\ref{mainthm}, we deduce a decay of discrepancies, valid for $s=d-2$, and which precisely shows that minimizers of the Coulomb jellium energy have a controlled discrepancy in all dimensions:
\begin{theorem}[Discrepancy bound of jellium minimizers]\label{discrepancybound}
Assume that $s= d-2$, that there exist constants $\underline m, \overline m>0$ such that $\underline m\le m_V(x)\le \overline m$ for all $x\in\op{supp}(m_V)$. Further assume that we are in the regime in which \eqref{mainlimit} holds and that $E_n'$ satisfy the charge separation condition of Proposition~\ref{separation}. Then letting 
\[
\nu_n':=\sum_{i=1}^n \delta_{x_i'}
\]
we have a finite asymptotic bound of the discrepancy of the $\nu_n'$ with respect to $\mu_V'$ as follows:
\begin{equation}\label{discrepdecay}
\lim_{\eta\to 0}\limsup_{\ell\to\infty}\limsup_{n\to\infty}\frac{1}{\ell^{d-1}}\left|\nu_n'(K_\ell(a)) - \int_{K_\ell(a)} m_V'(x)dx\right|<\infty.
\end{equation}
\end{theorem}
We note that for the $d=2$ case the above result is already present in \cite{rns}. A weaker version in which, still for $d=2$, the decay of the absolute value term in \eqref{discrepdecay} is shown to be $o(\ell^d)$ rather than $O(\ell^{d-1})$ like here, was also proved via Beurling-Landau densities, in \cite{aoc}.\\

It is interesting to compare the discrepancy bound \eqref{discrepdecay} with the notion of hyperuniform \emph{random} point configurations \cite{torstil, zacharytorquato, torquato2016hyperuniformity}. A random point configuration is called \emph{hyperuniform} if the variance $\sigma^2(\ell)$ of the number of points present in a window $K_\ell(x)$ of size $\sim \ell$ grows like a surface term, i.e. $\sigma^2(\ell)\sim \ell^{d-1}$. In our case, a notion of local number variance still can be defined by considering the number $N_{\ell,n}:=\nu_n'(K_\ell(n^{1/d}x))$ as a random variable on the probability $(\Omega, \mathcal B, P)$ where $\Omega=\op{supp}(\mu_V)$, $\mathcal B$ are the Borel sets of $\Omega$ and $P$ is the uniform measure on $\Omega$. In this case we are tempted to conjecture that configuratios are hyperuniform in the sense that asypmptotically in the regime $1\ll\ell\ll n$ the variance of $N_{\ell,n}$ grows like $O(\ell^{d-1})$. Proving this however will require to face new difficulties in order to produce a quantitative study of two-point density correlations analogous to \cite{torstil} for sequences of minimizers, and we leave this endeavor to future work.
\\

Our above control of discrepancies could also prove useful in making more rigorous the study of \emph{scars}, \emph{i.e.} the study of topological defects appearing in numerical simulations of point configurations on manifolds. In this case it is apparent by numerical simulations that defects arise, with a large literature focussing on the case of points distributed on the $2$-dimensional sphere \cite{bnt, scars0, scars1,scars2} and also on more general surfaces and manifolds as in \cite{scars4,scars3}. However in this case rigorous mathematical studies of the asymptotics of defects in the large-$n$ limit seem to be difficult, also due to the lack of a well-accepted and easy to control notion of ``defect''. In the case of the $2$-dimensional sphere, a first step in the study of the next-order term which may allow to obtain a version of our functional $\mathcal W$ on the sphere has been recently provided by B\'etermin-Sandier in \cite{betermin2014renormalized}. As the presence of scars may be detected by the presence of localized higher discrepancy regions, it seems that the above result, if transferred to the case of points living on compact manifolds, may provide a tool towards a rigorous proof of such asymptotics.

\section{Assumptions and main definitions}\label{assumptionsdefinitions}
\subsection{Hypotheses on $V$ and $\mu_V$}
The minimization of $I$, defined by \eqref{defI}, on $\mathcal P(\mathbb R^d)$, the space of probability measures on $\mathbb R^d$, is a standard problem in potential theory (see \cite{frostman} or \cite{safftotik} for $d=2$). In particular, if $V$ satisfies the following assumptions:
\begin{eqnarray}\label{assv1}
&V \mbox{ is l.s.c. and bounded below},\\ \label{assv2}
&\{x: V(x)<\infty\} \mbox{ has positive $g$-capacity},\\ \label{assv3}
&\lim_{|x|\to\infty} V(x)=+\infty, \quad \text{resp.} \  \lim_{|x|\to \infty} \frac{V(x)}{2}- \log |x|= + \infty \  \text{in cases }\eqref{wlog}-\eqref{wlog2d},
\end{eqnarray} 
then the minimum of $I$ over $\mathcal P(\mathbb R^d)$ exists, is finite and is achieved by a unique equilibrium measure $\mu_V$, which has a compact support $\Sigma$ of positive $g$-capacity. In addition $\mu_V$ is uniquely characterized by the fact that\footnote{Recall that using the usual notation of potential theory \cite{landkof}, here ``quasi everywhere'', abbreviated ``q.e.'', means ``up to sets of zero $g$-capacity''.}
\[
  \left\{\begin{array}{ll} h^{\mu_V} +\tfrac{V}{2} \ge c& \text{ q.e. }\ , \\ [2mm]
          h^{\mu_V}+\tfrac{V}{2}=c&\text{ q.e. on } \Sigma.
         \end{array}\right.
 \]
 where  $h^{\mu_V}(x):=\int g(x-y) d\mu_V(y)$ and $c:=I(\mu_V) - \int\tfrac{V}{2}d
 \mu_V$. 
 \par Note that $h^{\mu_V}$ can be characterized as the unique solution of a fractional obstacle problem, and the corresponding regularity theory \cite{sil,cs,crs} allows to obtain regularity results on $h^{\mu_V}$ and on the free-boundary $\partial \Sigma$ in terms of the regularity of $V$. We will write 
 
\begin{equation}\label{defzeta}
\zeta:=h^{\mu_V} +\tfrac{V}{2} - c\ge 0.
\end{equation}
Like in  \cite{ss2d,rs,ps}, it is assumed that $\mu_V$ is absolutely continuous with respect to the Lebesgue measure, with density also denoted by $m_V$, and in order to make the explicit constructions easier, we need to assume that this density is bounded and sufficiently regular on its support. More precisely, we make the following technical, and certainly not optimal, assumptions: 
\begin{eqnarray}\label{assumpsigma}
& \partial \Sigma\  \text{is} \ C^1 \\\label{assmu1}
& \text{$\mu_V$  has a density which is $ C^{0, \beta}$ in $\Sigma$, } \\\label{assmu2}
& \exists c_1, c_2, \overline m>0\text{ s.t.  }c_1\op{dist} (x, \partial \Sigma)^\alpha \le m_V(x) \le \min(c_2 \op{dist}(x, \partial \Sigma)^\alpha,  \overline m)<\infty \text{  in }\Sigma, \nonumber
\end{eqnarray}
with the conditions 
\begin{equation}
\label{condab}
0<\beta \le 1,\qquad 0\le \alpha \le \frac{2\beta d}{2d-s}.
\end{equation}
Of course if $\alpha<1$ one should take $\beta=\alpha$, and if $\alpha \ge 1$, one should take $\beta=1$ and $\alpha \le \frac{2d}{d-s}$.
These assumptions include the case of the semi-circle law  $\frac{1}{2\pi} \sqrt{4-x^2}1_{|x|<2}$ 
arising for the quadratic potential in \eqref{wlog}. In the Coulomb cases, a quadratic potential gives rise to the equilibrium  measure $c 1_B$ where $B\subset\mathbb R^d$ is a ball, a case also covered by our assumptions with $\alpha=0$. In the Riesz case, any compactly supported radial profile can be obtained 
as the equilibrium measure associated to some potential (see \cite[Corollary 1.4]{CGZ}). Our assumptions are thus never empty.

\subsection{Blowup, regularization and splitting formula}

The renormalized energy appears in \cite{ps} as a next order limit of $H_n$ after a blow-up is performed, at the inverse of the typical nearest neighbor distance between the points, \emph{i.e.} $n^{1/d}$. It is expressed in terms of the potential generated by the configuration $x_1, \cdots, x_n$ and defined by 
\begin{equation}
\label{defhn}
h_n(X)=  g * \left(\sum_{i=1}^n \delta_{(x_i,0)} - n {m_V}\delta_{\mathbb R^d}\right).
\end{equation}
Recall that the Riesz kernel $g$ is not the convolution kernel of a local operator, as in the Coulomb case $s=d-2$ or \eqref{wlog2d}, where $g$ is the kernel of the Laplacian. It is the kernel of a fractional Laplacian, which is a nonlocal operator. It turns out however that, as originally noticed in \cite{cafsil}, if $d-2<s<d$ then this fractional Laplacian operator can be transformed into a local but inhomogeneous operator of the form $\mathrm{div}(\yg \nabla\cdot)$ by adding one space variable $y\in \mathbb R$ to the space $\mathbb R^d$. The number $\gamma$ is chosen such that
\begin{equation}\label{choicegamma}
\gamma=s-d+2-k
\end{equation}
where $k$ will denote the dimension extension. We will take $k=0$ in all the Coulomb cases, \emph{i.e.} $s=d-2$ and $d\ge 3$ or \eqref{wlog2d}. In all other cases, we will need to take $k=1$. In the particular case of $s=d-1$ then $\gamma=0$, and this correspond to using a harmonic extension (see \cite{cafsil,ss1d,ps} for more details). Points in the space $\mathbb R^d$ will be denoted by $x$, and points in the extended space $\mathbb R^{d+k}$ by $X$, with $X=(x,y)$, $x\in \mathbb R^d$ and $y\in \mathbb R^k$.

\par For the blown-up quantities  we will use the following notation (with the convention $s=0$ in the cases \eqref{wlog} or \eqref{wlog2d}):
\begin{align}
& x'= n^{1/d} x \quad X'=n^{1/d} X \quad x_i'=n^{1/d} x_i\\ 
&  m_V'(x')= m_V(x)\\
& h_n'(X')= n^{-\frac{s}{d}} h_n(X).
\label{rescalh}
\end{align}
In particular if $\Sigma=\op{supp}(m_V),\Sigma'=\op{supp}(m_V')$ then there holds 
\begin{equation}
\label{supports}
\Sigma'=n^{\frac{1}{d}}\Sigma.\end{equation}
We note that $h_n$ and  $h_n'$ satisfy
\begin{equation}
\label{hn}
-\op{div} (\yg \nabla h_n) = c_{s,d} \Big( \sum_{i=1}^n \delta_{x_i} - n m_V \delta_{\mathbb R^d}\Big) \quad \text{in } \ \mathbb R^{d+k}\, , \end{equation}
\begin{equation}
\label{hnp}
-\op{div} (\yg \nabla h_n') = c_{s,d} \Big( \sum_{i=1}^n \delta_{x_i'} - m_V' \delta_{\mathbb R^d}\Big) \quad \text{in } \ \mathbb R^{d+k}\, ,  \end{equation}

In order to define our renormalized energy we need to truncate and regularize the Riesz (or logarithmic) kernel. We define the truncated kernel as in \cite{ps}, in the following way: for  $1>\eta>0$ and  $X\in \mathbb R^{d+k}$, let 
\begin{equation}
\label{feta} 
f_\eta(X)= \left(g(X)- g(\eta)\right)_+.\end{equation}
We note that the function  $f_\eta$ vanishes outside of $B(0, \eta)$ and satisfies that
 \begin{equation}
\label{defde}
\delta_0^{(\eta)}:= \frac{1}{c_{s,d}}\op{div} (\yg \nabla f_\eta)+ \delta_0
\end{equation}
 is a positive  measure supported on $\partial B(0, \eta)$, and which is such that  for any test-function $\varphi$, 
\begin{equation*}
\int \varphi \delta_0^{(\eta)}=\frac{1}{c_{s,d}} \int_{\partial B(0, \eta)}  \varphi (X)\yg g'(\eta) .
\end{equation*}
One can check that $\delta_0^{(\eta)}$ is the uniform measure of mass $1$ on $\partial B(0, \eta)$, and we may write 
\begin{equation}
\label{divf}
-\op{div} (\yg \nabla f_\eta) = c_{s,d} ( \delta_0 - \delta_0^{(\eta)})\quad \text{in} \ \mathbb R^{d+k}.
\end{equation}
We will also denote by $\delta_p^{(\eta)}$ the measure $\delta_0^{(\eta)} (X-p)$, for $p \in \mathbb R^d\times \{0\}$. 
In the Coulomb cases, \emph{i.e.} when $k=0$, then $\delta_0^{(\eta)}$ is the same as in \cite{rs}. If $h$ can be written in the form \eqref{defhn}, then we will also denote 
\begin{equation}
\label{defheta}
h_\eta:= h- \sum_{i=1}^n f_\eta (x-x_i).
\end{equation} 
\begin{remark}\label{rmktrunc}
For $h=h_n$ as in \eqref{defhn} the transformation from $h$ to $h_\eta$ amounts to truncating the kernel $g$, but only for the Dirac part of the r.h.s. Indeed, letting $g_\eta(x)=\min (g(x), g(\eta))$ be the truncated kernel, we have 
$$h_\eta= g_\eta* (\sum_{i=1}^n\delta_{x_i}) - g* (m_V\delta_{\mathbb R^d}).$$
\end{remark}
 In view of \eqref{divf}, $h_{n, \eta}$ and $h_{n, \eta}'$, defined from  $h_n$ and $h_n'$ via \eqref{defheta}, satisfy
\begin{equation}
\label{hne}
-\op{div} (\yg \nabla h_{n, \eta}) = c_{s,d} \Big( \sum_{i=1}^n \delta_{x_i}^{(\eta)} - n m_V \delta_{\mathbb R^d}\Big) \quad \text{in } \ \mathbb R^{d+k}\, ,  \end{equation}
\begin{equation}
\label{hnpe}
-\op{div} (\yg \nabla h_{n, \eta}') = c_{s,d} \Big( \sum_{i=1}^n \delta_{x_i'}^{(\eta)} - m_V' \delta_{\mathbb R^d}\Big) \quad \text{in } \ \mathbb R^{d+k}\, ,  \end{equation}
with the usual embedding of $\mathbb R^d $ into $\mathbb R^{d+k}$. We now recall the splitting formula from \cite{ps}. 
\begin{proposition}[Splitting formula]\label{splitting}
For any $n$, any $x_1, \cdots, x_n$ distinct points in $  \mathbb R^d\times\{0\}$, letting $h_n$ be as in \eqref{defhn} and  $h_{n,\eta}$ deduced from it via   \eqref{defheta}, we have in the case \eqref{kernel}
\begin{multline}
\label{propsplit}
H_n(x_1, \cdots, x_n) = n^2 I (\mu_V) +2 n\sum_{i=1}^n \zeta(x_i)
+  n^{1+\frac{s}{d}} \lim_{\eta\to 0} \frac{1}{c_{s,d}} \left( \frac{1}{n}\int_{\mathbb R^{d+k}}\yg |\nabla h_{n,\eta}'|^2 -c_{s,d}  g(\eta)\right),\end{multline} 
respectively in the cases \eqref{wlog}--\eqref{wlog2d}
\begin{multline}
H_n(x_1, \cdots, x_n) = n^2 I (\mu_V) +2 n\sum_{i=1}^n \zeta(x_i) - \frac{n}{d}\log n
+ n  \lim_{\eta\to 0}\frac{1}{c_{s,d}} \left( \frac{1}{n }\int_{\mathbb R^{d+k}}\yg |\nabla h_{n,\eta}'|^2 - c_{s,d} g(\eta)\right).\end{multline}
\end{proposition}
One expects the repelling points $x_i$ to organise in a very uniform way, and thus that the interpoint distance asymptotically decreases like $n^{-1/d}$. The following is proven in \cite{ps}, by potential-theoretic methods \cite{landkof, bhs2} and using the maximum principle.
\begin{proposition}[Point separation, {\cite[Thm. 5]{ps}}]\label{separation}
Let $(x_1, \dots, x_n)$ minimize $H_n$. Then for each $i\in [1,n]$, $x_i \in \Sigma$, and for each $i\neq j$, it holds
$$|x_i-x_j| \ge \frac {r}{(n \max_x |m_V(x)|) ^{1/d}}, $$ where $r$ is some positive constant depending only on $s$ and $d$. 
\end{proposition}
 The scale $\sim n^{-1/d}$ is then termed the \emph{microscopic scale} of our gases, and the two-scale reformulation of the energy $H_n$ as done in \cite{ss2d, rs, ps} involves separating the energy contributions from the macroscale and from this microscopic scale. In particular the distribution of points at the microscopic scale is governed by the renormalized energy $\mathcal W$ to be introduced below.
\subsection{The renormalized energy}\label{secrenoren}
Consider the formulas appearing in Proposition~\ref{splitting}. As $\zeta\ge 0$ and $\zeta=0$ on $\op{supp}(\mu_V)$, it acts as an effective potential, favouring the configurations where the points $x_i$ are in the support of $\mu_V$. The last term produces the next-order term of the energy, and justifies the definition of the renormalized energy $\mathcal W$ of an infinite configuration of points. This functional $\mathcal W$ is defined via the gradient of the potential generated by the point configuration, embedded into the extended space $\mathbb R^{d+k}$. That gradient is a vector field that we denote $E$ (like electric field, by analogy with the Coulomb case). Then $E$ will solve a relation of the form 
\begin{equation}
\label{eqe} 
-\op{div} (\yg E) = c_{d,s} \Big(\sum_{p \in \Lambda} \delta_p - m(x)\delta_{\mathbb R^d}\Big) \quad \text{in} \ \mathbb R^{d+k}.
\end{equation}
where $\Lambda$ is some discrete set in $\mathbb R^d \times \{0\}$ (identified with $\mathbb R^d$). Due to the fact that (as recalled in Proposition~\ref{separation}) the minimizers of our energy have separated charges, we restrict in the present work to fields $E$ corresponding to multiplicity-one charges, as opposed to general positive integer multiplicity case considered in \cite{rs, ps}. For any such $E$ (defined over $\mathbb R^{d+k}$ or over subsets of it), we define, by a formula generalizing \eqref{defheta},
\begin{equation}
\label{defeeta}
E_\eta :=  E- \sum_{p \in \Lambda } \nabla f_\eta (X-p).
\end{equation}
We will write $\Phi_\eta$ for the map that sends $E$ to $E_\eta$, and note that it is a bijection from the set of vector fields satisfying a relation of the form \eqref{eqe} to those satisfying a relation of the form 
\begin{equation}
\label{eqeeta} 
-\op{div} (\yg E_\eta) = c_{d,s} \Big(\sum_{p \in \Lambda}\delta_p ^{(\eta)}- m(x)\delta_{\mathbb R^d}\Big) \quad \text{in} \ \mathbb R^{d+k}.
\end{equation}
The class of vector fields on which we are going to concentrate is thus the following:
\begin{definition}[Admissible vector fields]
\label{defbam} 
Given a non-negative density function $m:\mathbb R^d\to\R^+$, we define the class $\mathcal{A}_m$ to be the class of gradient vector fields $E=\nabla h$ that satisfy
\begin{equation}\label{eqclam}
-\op{div} (\yg E) = c_{s,d} \Big( \sum_{p \in \Lambda} \delta_p-m(x)\delta_{\mathbb R^d} \Big) \quad \text in\ \mathbb R^{d+k}
\end{equation}
where $\Lambda $ is a discrete set of points in $\mathbb R^d\times \{0\}$.
\end{definition}
In case $m\in L^\infty_{loc}$, vector fields as above 
 blow up exactly in $1/|X|^{s+1}$ near each $p \in \Lambda$
 (with the convention $s=0$ for the cases \eqref{wlog}--\eqref{wlog2d}); 
 such vector fields naturally belong to the space $L^p_{loc}(\mathbb R^{d+k}, \mathbb R^{d+k})$ for $p<\frac{d+k}{s+1}$.

We are now in a position to define the renormalized energy for blow-up configurations like in \cite{rs,ps}. 
In the definition, we let $K_R$ denote the hypercubes $[-R/2,R/2]^d.$

\begin{definition}[Renormalized energy]
\label{renerg}
Let $E\in \mathcal{A}_m$ satisfy \eqref{eqe} and $f:\mathbb R^{d+k}\to \R^+$ be a measurable function. Then for $0<\eta<1$ we define 
\begin{equation}\label{Weta}
\mathcal W_\eta(E,f)=\int_{\mathbb R^{d+k}}\yg f |E_\eta|^2 - c_{d,s}g(\eta)\int_{\mathbb R^{d+k}}f\sum_{p\in\Lambda}\delta_p^{(\eta)}. 
\end{equation}
For $A\subset\mathbb R^d$ a Borel set we define $\mathcal W_\eta (E,A):=\mathcal W_\eta(E,1_{A\times\R^k})$ where $1_S$ is the characteristic function which equals $1$ on a set $S$ and $0$ outside $S$. We then define 
\begin{equation}\label{defW}
\mathcal W(E,A) = \lim_{\eta\to 0} \mathcal W_\eta(E,A),\qquad \mathcal W(E) = \lim_{\eta\to 0} \limsup_{R\to\infty}\frac{\mathcal W_\eta(E, K_R)}{R^d}.
\end{equation} 
\end{definition}

\begin{remark}
\label{smoothnonsmooth}
Note that if $\chi_{A,\epsilon}(x,y)=1_A * \rho_\epsilon^{(d)} (x)1_{[-R_\epsilon, R_\epsilon]^k}*\rho_\epsilon^{(k)}(y)$ are $C^\infty_c$ functions approximating $1_{A\times\R^k}$ where $R_\epsilon\to\infty$ as $\epsilon\to 0$ and $\rho_\epsilon^{(n)}(z)=\epsilon^{-n}\rho^{(n)}(z/\epsilon)$ denotes mollifiers based on a smooth radial probability density $\rho^{(n)}$ supported on the unit ball of $\R^n$ then we have $\mathcal W_\eta(E,A)=\lim_{\epsilon\to 0}\mathcal W_\eta(E,\chi_{A,\epsilon})$, by monotone convergence in \eqref{Weta}.
\end{remark}

The name renormalized energy (originating in Bethuel-Brezis-H\'elein \cite{bbh} in the context of two-dimensional Ginzburg-Landau vortices) reflects the fact that $\int \yg |\nabla h|^2 $ which is infinite, is computed in a renormalized way by first changing $h$ into $h_\eta$ and then removing the appropriate divergent terms $c_{s,d}g(\eta)$ corresponding to all points.

The above is a generalized version of the renormalized energy defined in \cite{ps}, and fits in the framework of the study of ``jellium energies'', for which we refer to \cite{blanclewin} and to the references therein. As in \cite{rs,ps} the next-order functional $\mathcal W$ differs from the one defined in previous works by Sandier-Serfaty \cite{gl13,ss1d} for the one and two-dimensional logarithmic interaction, essentially in the fact that the order of the limits $R\to \infty $ and $\eta\to 0$ is reversed. We refer to \cite{rs} for a further discussion of the comparison between the two.

In the case of constant $m$, by scaling we may always reduce to studying the class $\mathcal A_1$, indeed, 
if $E\in \mathcal A_m$ and $A$ is Borel, then $\hat{E} = m^{-\frac{s+1}{d}}   E(c_{s,d}\cdot m^{-1/d})\in \mathcal A_1$ \footnote{with the convention $s=0$ in the case \eqref{wlog}} and 
\begin{equation}
\label{scalingW}
\mathcal W_\eta(E,A)=m^{1+s/d} \mathcal W_{\eta m^{1/d}}  (\hat{E}, m^{1/d}A)  \qquad \mathcal W(E)= m^{1+s/d} \mathcal W(\hat{E}).
\end{equation} 
in the case \eqref{kernel}, and respectively
\begin{equation}
\label{scalinglog}
\mathcal W_\eta(E,A)= m\left(\mathcal W_{m\eta} (\hat{E}, m^{1/d}A)-\frac{2\pi}{d} \log m\right)\qquad
\mathcal W(E)= m\left(\mathcal W(\hat{E})-\frac{2\pi}{d} \log m\right)
\end{equation}
in the cases \eqref{wlog}--\eqref{wlog2d}.

\subsection{Discussion on the hypothesis~\eqref{hypdecayfields}}\label{discussextrahyp}

\subsubsection{Power-law bound}
For studying the case $k=1$ we need a further assumption regarding decay in the extra dimension (\emph{i.e.} as $|y|\to\infty$) of the energy vector fields $E_n'$ corresponding to minimizing configurations of points. It is tempting to conjecture that for fields $E'_n$ corresponding to minimizers there exist constants $a, C_2>0$ such that for each bilipschitz deformation of a cube $K\subset\mathbb R^d$ there holds 
\begin{equation}\label{decrvert} 
\int_{K\times (\mathbb R\setminus[-t,t])}|y|^\gamma |E_n'|^2 \le C_2|K|t^{-a}.
\end{equation}
We use here the notation $C_2$ in order to allow the reader to directly compare the bounds here with those from Propositions \ref{goodbdry} and \ref{screening} below.
Such decay is not true for general configuratons for which $\mathcal W(E)<\infty$, and it seems to be equivalent to a uniformity condition on the field-generating configurations. We note here that this condition holds in the case of lattice-like configurations:
\begin{lemma}[Power-like decay for lattices]
Assume that $\Lambda\subset\mathbb R^d$ is a Bravais lattice of density one and consider the admissible vector field $E'$ corresponding to Definition \ref{defbam} with this choice of $\Lambda$ and measure $m\equiv 1$. Then \eqref{decrvert} holds with $a=s-d+2>0$.
\end{lemma}
\begin{proof}
We prove the result in the case $\Lambda=\mathbb Z^d$ but the same proof works for a general Bravais lattice of density one. We can calculate the norm of $E_{loc}:=\nabla h=\nabla(g * (\sum_{p \in \Lambda} \delta_p - \delta_{\mathbb R^d}))$ at a point $(x_0,y)$ with $|y|=t$. As $\Lambda$ is periodic, so is $E_{loc}$ and we may suppose that $x_0\in[0,1]^d$. Then we find, 
\begin{eqnarray*}
 h(x_0,y)&=& -|\mathbb S^{d-1}|\int_0^\infty \frac{r^{d-1}} {(t^2+r^2)^{s/2}} dr + \sum_{p\in\mathbb Z^d+x_0} \frac{1}{(t^2+|p|^2)^{s/2}}\\
 &=&-t^{d-s}|\mathbb S^{d-1}|\int_0^\infty \frac{\rho^{d-1}d\rho}{(1+\rho^2)^{s/2}} + t^{-s} \sum_{q\in t^{-1}(\mathbb Z^d+x_0)}\frac{1}{(1+|q|^2)^{s/2}},
\end{eqnarray*}
and we see that this expression is $t^{-s}$ times the error in the approximation for the Riemann integral of $(1+|x|^2)^{-s/2}$ given by the partition of $\mathbb R^d$ in cubic cells centered at $t^{-1}(\mathbb Z^d+x_0)$. To compute $\nabla h= \nabla (g *(\sum_{p \in \Lambda} \delta_p - \delta_{\mathbb R^d}))$ at the same point we must take into account the fact that cancellations occur: for the continuum counterpart the horizontal contributions to $\nabla h$ from points symmetric with respect to $x$ cancel, and the length of the vertical component is $t/(t^2+r^2)^{1/2}$. We thus find that $|E_{loc}|(x,t)$ is $t^{-s-1}$ times the Riemann sum approximation error for the integral of $|x|(1+|x|^2)^{-(s+3)/2}$ corresponding to the decomposition of $\mathbb R^d$ by cubic cells centered at $t^{-1}(\mathbb Z^d+x_0)$. It follows that uniformly in $x_0$ we have the sharp power decay bound $|\nabla h|\le t^{d-s-2}$ and thus (using also the relation $\gamma = s-d+2-k=s-d+1$ valid here)
\[
\int_{\mathbb R\setminus [-t,t]} |y|^\gamma |E_{loc}|^2(x_0,y) dy\le C t^{2d-2s-3+\gamma}= C t^{d-2-s},
\]
which implies \eqref{decrvert}.
\end{proof}

\subsubsection{Weaker decay bound}
In general we were not able to find a way to prove a bound of the form \eqref{decrvert} for our minimizers $E_n'$, and we leave the question of whether the lattice-configurations have the same decay power $a$ as the minimizers for future work. This can be seen as a uniformity conjecture on the minimizers, which therefore is a weaker version of the conjecture that lattices are minimizers.

 Note also that so far we didn't exclude that there exist distinct minimizing configurations with different decay of the energy. 
 %However we can prove compactness results of a kind which suffice for the present endeavors.
 
 We will denote, compatibly with the notation \eqref{decrvert}, that for a vector field $E$, a cube $K\subset \mathbb R^d$ and a ``height" $t$,
\begin{equation}\label{defc2tk}
 \frac{1}{|K|}\int_{K\times(\mathbb R\setminus[-t,t])}|y|^\gamma |E|^2 := C_2(E,t,K).
\end{equation}
In the above notation \eqref{decrvert} states that $C_2(E,t,K)= C_2 t^{-a}$, independently of $K$. We denote as follows some more global bounds, provided that they are finite.
\begin{eqnarray*}
C_2(E,t,L)&:=&\sup\left\{C_2(E,t,K_l(a)):\ l\in[L/2,2L],\ a\in\mathbb R^d\right\},\quad L>0,\\
C_2(E,t)&:=&\sup\left\{C_2(E,t,L): L >1\right\}.
\end{eqnarray*}
We then see that, by expressing the average over $K_{2l}(a)$ as the average of the $2^d$ averages over distinct subcubes $K_l(a')$ partitioning $K_{2l}(a)$ up to measure zero, we find
\begin{equation}\label{mediadimedie}
C_2(E,t,2L)\le C_2(E,t,L) \quad\text{ for all }\quad t>0.
\end{equation}
Recall that in \cite[Prop. 7.1]{ps} it was proved that there exists a minimizer $E$ of $\mathcal W_\eta$ over $\mathcal A_1$ which satisfies
\[
\lim_{t\to\infty}\lim_{R\to\infty}C_2(E,R,t) =0,
\]
and moreover (see \cite[Sect. 5]{ps}) a weak limit of a subsequence of rescaled minimizers $E_n'$ provides an $E$ with the above property for generic choices of the blow-up centers.

We need however to use the stronger fact that the choice of $E_n'$ corresponding to a blow-up sequence of minimizers in our problem, satisfies this type of uniform bound too as stated in hypothesis \eqref{hypdecayfields}. 

\section{Screening lemmas}

The following is the main tool for selecting the good boundaries in our constructions. This is used later in combination with our precise splitting formula of Proposition~\ref{prodecr} in performing the screening construction.
\begin{proposition}[Good boundary slices]\label{goodbdry}
Let $K_T$ be a rectangle of sidelenghts $\sim T$, let $\rho\in L^\infty(K_{T})$ and fix $\eta\in]0,1[$. If $\Lambda\subset\mathbb R^d$ is a discrete set, we define
\[
\nu=\sum_{p\in\Lambda}\delta_p,\quad \nu^{(\eta)} = \sum_{p\in\Lambda} \delta_p^{(\eta)}.
\]
Assume that $E=\nabla h$ as in \eqref{hnp} on the rectangle $K_T$, in particular
\[
-\op{div}(\yg E_\eta)=c_{d,s}\left(\nu^{(\eta)} - \rho(x)\delta_{\mathbb R^d}\right) \quad \text{ in }K_{T}\times\R^k.
\]
Further assume that $E$ is controlled in the following sense: 
\begin{equation}\label{bdsgbdry1}
       \frac{1}{T^d}\int_{K_T\times[-t, t]^k}|y|^\gamma|E_\eta|^2\le C_1
\end{equation}
for some positive constant $C_1$ and $t\in [0,T]$.

Let $\ve_1\in]0,1[, L_i\in [\ve_1^{1/d} T_i, T_i]$, $i=1,\ldots,d$ and $a\in \mathbb R^d\times \{0\}$ such that the rectangle $K_L(a)$ of sidelenghts $L_i$ is included in $K_T$ and, if $k=1$, assume that 
\begin{equation}\label{bdsgbdry2}
       \frac{1}{L^d}\int_{K_L\times\left(\mathbb R\setminus[-\frac12 t, \frac12 t]\right)}|y|^\gamma|E_\eta|^2\le C_2.
\end{equation}
for positive constant $C_2$ and $t\in [0,T]$.

Let $l\in ]0, L/3]$. There exists a rectangle $K_L'(a)\subset K_L(a)$ such that $\op{dist}(\partial K_L'(a), \partial K_L(a))\in[l, 2l[$ and a universal constant $C$ such that the following hold: 
\begin{eqnarray}
  \int_{\partial K_L'\times[-t, t]^k}\yg|E_\eta|^2 &\le& C C_1 \ve_1^{-1} l^{-1}L^d\ ,\label{enbd1}
\end{eqnarray}
and when $k=1$ there exists $t'\in [t/2, t]$ such that
\begin{equation}\label{decv}
\int_{K_L' \times \{-t', t'\}} \yg |E_\eta|^2 <  C C_2L^{d}t^{-1}.
\end{equation}
\end{proposition}
\begin{remark}\label{sepcariche}
We recall that if the charges corresponding to a vector field $E$ are well-separated at distance $r_0$ then the value of $\mathcal W_\eta(E, A)$ over a domain $A$ is bounded in terms of the volume $|A|$. More precisely, if $C_{r_0}=|B_{r_0}|^{-1}$ then we have
\begin{equation}\label{boundsepcariche}
\int_{A\times\mathbb R^k}\yg|E_\eta|^2 - C_{r_0}c_{d,s}g(\eta)|A|\le\mathcal W_\eta(E,A)\le\int_{A\times\mathbb R^k}\yg|E_\eta|^2\ ,
\end{equation}
therefore in our error bounds on thin boundary layers of cubes the main difficulty is in estimating the $L^2_{\yg}$-norm of $E_\eta$ rather than the remainder of $\mathcal W_\eta$. In the setting of Proposition~\ref{goodbdry} we obtain the bound
\begin{equation}\label{boundcubisepcariche}
\frac{\mathcal W_\eta(E, K_L)}{|K_L|}\ge\frac{\mathcal W_\eta(E, K_L')}{|K_L'|} + C_{s,d,\overline m}(g(\eta)+1)\frac{l}{L}, 
\end{equation}
where we include in the notation the fact, mentioned in Proposition~\ref{separation} (see \cite[Thm. 5]{ps}) that $r_0$ depends on $\max_x m(x), s, d$ only.
\end{remark}
\begin{proof}
We consider the boundaries $\partial K_\tau(a)$ for $\tau\in[L-2l,L-l]$. Then the existence of $K_L'(a)$ such that \eqref{enbd1} holds, follows by applying the mean value theorem and using \eqref{bdsgbdry1} and the fact that $L_i\ge \ve_1^{1/d}T_i$. To obtain \eqref{decv} for $k=1$, we use a mean value principle and \eqref{bdsgbdry2} to find $t'\in[t/2,t]$ such that 
\[
\int_{K_L \times \{-t', t'\}} \yg |E|^2 <  2 C_2L^dt^{-1}, 
\]
after which \eqref{decv} follows.
\end{proof}

\subsection{Interchanging small energy boundary data by screening}
The following is a version of proposition 6.1 of \cite{ps} with a varying background measure instead of a constant one. We include the case where we pass with a small energy change from a good boundary datum on the inner hyperrectangle to zero on the outer one, and the case where such roles are inverted and we rather pass from a good boundary on the outer hyperrectangle to zero boundary datum on the inner one.

\begin{proposition}[Screening]\label{screening}

Let $K_T$ be a rectangle of sidelenghts $\sim T$, let $\rho$ be a $C^{0,\alpha}(K_{T+1})$ function with $\alpha>0$ for which there exist $ \underline \rho, \overline \rho>0$ such that $\underline \rho\le \rho(x) \le \overline \rho$. Fix $\eta\in]0,1[$, let $\Lambda\subset\mathbb R^d$ and assume there exists $r_0\in ]0,1/2]$ such that $\op{min}\{|p-q|:\ p\neq q \in \Lambda\} \ge \sqrt d r_0$. Assume that $E=\nabla h$ such that  
\[
-\op{div}(\yg E)=c_{d,s}\left(\sum_{p\in\Lambda}\delta_p - \rho(x)\delta_{\mathbb R^d}\right) \quad \text{ in }K_{T}\times\R^k,
\]
and that $E$ is controlled as in~\eqref{bdsgbdry1} with $C_1>0$. Moreover, let $ \ve_1\in]0,1[,\ L_i\in [\ve_1^{1/d}T_i, T_i]$ for $i=1,\ldots,d$ and $a\in \mathbb R^d\times \{0\}$ be such that the rectangle $K_L(a)$ of sidelenghts $L_i$ is included in $K_T$. If $k=1$, assume also that $C_2$ is such that \eqref{bdsgbdry2} holds, \emph{i.e.} that $C_2\ge C_2(E,t,L)$. Let $K_L'(a), l$ be as in the claim of Proposition~\ref{goodbdry}. 

There exists a constant $c_0$ such that if $L$ is sufficiently large and
\begin{equation}\label{condLt}
\left\{
\begin{array}{cl}
C_1^{\frac{1}{2}}t^{-\frac{d+1}{2}}L^{\frac{d}{2}}\ve_1^{-\frac12} l^{-\frac12}\le c_0&\text{ for } k=0\\
t^{\frac{\gamma-3}{2}}LC_2^{\frac{1}{2}}\le c_0 \text{ and }C_1^{\frac{1}{2}}t^{\frac{-d+\gamma}{2}}L^{\frac{d}{2}}\ve_1^{-\frac12}l^{-\frac12}\le c_0 &\text{ for } k=1
\end{array}
\right.,
\end{equation}
and if we denote 
\begin{equation}\label{deferrsc}
err_{sc}(\eta,t,L,l,\ve_1,C_1, C_2):=C_1\ve_1^{-1}l^{-1}g(\eta)+C_2 t^{-1}L+C_1\ve_1^{-1}l^{-1}  t+g(\eta)L^{-1}t ,
\end{equation}
then the following hold.

\begin{enumerate}
\item There exists $\bar K_L(a)$ such that $K_L'(a)\subset \bar K_L(a)$, $\op{dist}(\partial K_L'(a), \partial \bar K_L(a))\in [t,2t]$ and $\int_{\bar K_L(a)}\rho(x)\,dx\in \mathbb{N}$. There exist a vector field $\tilde E\in L^p_\loc(\bar K_L(a)\times \R^k,\mathbb R^{d+k})$ with $p<\min(2, \frac{2}{1+\gamma}, \frac{d+k}{s+1})$ and a subset $\tilde \Lambda\subset \bar K_L(a)$ such that 
\[
\left\{
\begin{aligned}
&-\op{div}(\yg \tilde E)=c_{d,s}\left(\sum_{p\in\tilde \Lambda}\delta_p - \rho(x)\delta_{\mathbb R^d}\right) & \text{ in }\bar K_{L}(a)\times\R^k\\
&\tilde E\cdot\vec{\nu}=0  & \text{ on }\partial\bar K_{L}(a)\times\R^k\\
&\tilde E\cdot\vec{\nu}=E\cdot \vec{\nu}  & \text{ on }\partial K'_{L}(a)\times [-t,t]^k\\
&\tilde E=E , \ \tilde \Lambda=\Lambda & \text{ in } K'_{L}(a)\times [-t,t]^k\\
\end{aligned}
\right..
\]
Moreover,
\begin{align}\label{energybigger}
\frac{1}{L^d} \int_{(\bar K_L\setminus K_L')\times\mathbb R^k}\yg|\tilde E_\eta|^2 \le C err_{sc}(\eta,t,L,l,\ve_1, C_1, C_2).
\end{align}
Finally, the minimal distance between points in $\tilde \Lambda\setminus \Lambda$, and between points in $\tilde \Lambda\setminus \Lambda$ and $\partial (\bar K_L(a)\setminus K_L(a)')$ is bounded below by $r_1$ with $r_1$ a positive constant depending only on $d, \underline \rho, \overline \rho$.  

\item There exists $\underline K_L(a)$ such that $\underbar K_L(a)\subset K_L'(a)$, $\op{dist}(\partial K_L'(a), \partial \underbar K_L(a))\in [t,2t]$ and $\int_{\underline K_L(a)}\rho(x)\,dx\in \mathbb{N}$. There exist a vector field $\tilde E\in L^p_\loc( K_L(a)\times \R^k,\mathbb R^{d+k})$ with $p<\min(2, \frac{2}{1+\gamma}, \frac{d+k}{s+1})$ and a subset $\tilde \Lambda\subset K_L(a)$ such that 
\[
\left\{
\begin{aligned}
&-\op{div}(\yg \tilde E)=c_{d,s}\left(\sum_{p\in\tilde \Lambda}\delta_p - \rho(x)\delta_{\mathbb R^d}\right) & \text{ in }K_{L}(a)\times\R^k\\
&\tilde E\cdot\vec{\nu}=0  & \text{ on }\partial\underline K_{L}(a)\times\R^k\\
&\tilde E\cdot\vec{\nu}=E\cdot \vec{\nu}  & \text{ on }\partial K'_{L}(a)\times [-t,t]^k\\
&\tilde E=E , \ \tilde \Lambda=\Lambda & \text{ in } (K_L(a)\setminus K'_{L}(a))\times [-t,t]^k\\
\end{aligned}
\right..
\]
Moreover,
\begin{align}\label{energybigger2}
\frac{1}{L^d} \int_{(K_L'\setminus \underline K_L)\times\mathbb R^k}\yg|\tilde E_\eta|^2 \le C err_{sc}(\eta,t,L,l,\ve_1,C_1, C_2).
\end{align}
Finally, the minimal distance between points in $\tilde \Lambda\setminus \Lambda$, and between points in $\tilde \Lambda\setminus \Lambda$ and $\partial (K_L'(a)\setminus \underline K_L(a))$ is bounded below by $r_1$.  

\end{enumerate}
\end{proposition}

The above result is obtained precisely along the lines of Proposition 6.1 in \cite{ps}, except for a series of modifications needed in order to accommodate the nonconstant $\rho$. Before describing the modifications of the proof, we prove a hyperrectangle subdivision lemma which generalizes \cite[Lem. 6.3]{ps} to the case of a controlled background measure. To re-obtain Lemma 6.3 of \cite{ps} from the statement below, one must take $\rho\equiv 1$ and then scale the units of length by $m^{-1/d}$.
\begin{lemma}[Subdivision of a hyperrectangle]
 \label{subdivision}
 Let  $H=[0,\ell_1]\times\cdots\times[0,\ell_d]$ be a $d$-dimensional hyperrectangle of sidelengths $\ell_i$ and let $\rho:H\to \R^+$ be a function such that for two constants $\underline\rho, \overline\rho$ there holds $0<\underline\rho\le\rho(x)\le \overline\rho$ for all $x\in H$ and assume $\int_H\rho= I\in \mathbb N^*$. Assume that for all $i$, $\ell_i\ge2/\underline\rho$. Fix a face $F$ of $H$. Then  there is a partition of $H$ into $I$  subrectangles $\mathcal R_j$,  such that the following hold
 \begin{itemize}
  \item  all rectangles have volume $1$,   
  \item the sidelengths of each $\mathcal R_j$ lie in the interval $[2^{-d}\overline \rho^{-d}\underline \rho^{d-1} ,2^d \underline\rho^{-d}\overline\rho^{d-1}]$,
  \item all the $\mathcal R_j$'s which have a face in common with $F$ have the same  sidelength in the direction perpendicular to $F$. 
 \end{itemize}
\end{lemma}

\begin{proof}
We proceed by induction on the dimension. In case $d=1$ the mean value theorem allows to successively find values $0=p_0<p_1<\cdots<p_I=\ell_1$ such that $\int_{p_i}^{p_{i+1}} \rho=1$. Then the lengths $p_{i+1}-p_i$ lie in $[\overline \rho^{-1}, \underline \rho^{-1}]$ due to the bounds on $\rho$. In case $d\ge 2$ we may assume without loss of generality that the special face $F$ is the one on $\{x_d=0\}$. Define
\[
\tilde\rho :=\frac{\int_H\rho}{|H|}\in[\underline\rho,\overline\rho], \quad b:=\left\lfloor\ell_d^{-1}\int_H\rho\right\rfloor=\left\lfloor\tilde\rho|F|\right\rfloor, 
\]
where $\lfloor a\rfloor:=\max(\mathbb Z\cap]-\infty,a])$, and partition $[0,\ell_d]$ into a certain number $a$ of segments $S_\alpha$ of $\rho$-mass $b$ and one extra segment $S_0$ of $\rho$-mass belonging to $[b,2b]$. Let $\bar\ell_\alpha$ be the lengths of these segments 
. We have $\bar\ell_\alpha|F|\underline\rho\le b\le\bar\ell_\alpha|F|\overline\rho$ so 
\[ 
\bar\ell_\alpha\in\frac{b}{|F|}[\overline\rho^{-1},\underline\rho^{-1}]=\frac{\lfloor\tilde\rho|F|\rfloor}{|F|}[\overline\rho^{-1},\underline\rho^{-1}]\subset\left[\frac{\underline \rho 
-\frac{1}{|F|}}{\overline\rho}\ ,\ \frac{\overline\rho}{\underline\rho}\right].
\]
Since $\ell_i\ge 2/\underline\rho$, we have $\underline\rho-|F|^{-1}\ge \frac12\underline\rho$ and thus
\begin{equation}\label{boundbarlalpha}
\bar\ell_\alpha\in\left[\frac12\frac{\underline\rho}{\overline\rho}\ ,\ \frac{\overline\rho}{\underline\rho}\right].
\end{equation}
Let $H_\alpha'$ be the hyperrectangles $[0,\ell_1]\times\cdots\times[0,\ell_{d-1}]\times S_\alpha$. Then pushing forward by the projection onto $F$ the restrictions $\mu\llcorner H_\alpha'$ gives measures $\mu_\alpha=\rho_\alpha dx$ of integer total mass equal to $b$. Note that we have bounds, $\rho_\alpha(x)\in\bar\ell_\alpha^{-1}[\underline\rho,\overline\rho]$. 
Now we apply the inductive hypothesis on $F$ with the measure $\rho_\alpha$ and we obtain a subdivision of $F$ into $(d-1)$-dimensional unit $\rho_\alpha$-mass hyperrectangles $\mathcal R_j'$. Then by definition of $\rho_\alpha$ the hyperrectangles $\mathcal R_j'\times S_\alpha$ are of unit $\rho$-mass. The sidelenghts of the $\mathcal R_j'$ then lie in 
\begin{equation}\label{boundotherlengths}
\bar\ell_\alpha[2^{1-d}\overline \rho^{1-d}\underline \rho^{d-2} ,2^{d-1} \underline\rho^{1-d}\overline\rho^{d-2}],
\end{equation}
which via \eqref{boundbarlalpha} gives the correct bound as in the thesis. We may perform the same procedure also for the last segment $S_0$, the only difference being that we proceed with a different value $\tilde b\in[b,2b[$ instead of $b$. The only effect of this change is that the bound \eqref{boundbarlalpha} on $\bar\ell_0$ is perturbed by a factor belonging to $[1,2[$, however combined with the bound \eqref{boundotherlengths} in the case $\alpha=0$, this still gives estimates within the range allowed in the thesis.
\end{proof}
\begin{proof}[How to modify the argument in \cite{ps} Section 6 to obtain the proof of Proposition~\ref{screening}:]
We consider only the first case where we desire to change our charges and vector fields only over $K_L'\subset \bar K_L$, which corresponds to the same setting as in \cite[Prop. 6.1]{ps}. The modifications needed for the case where we desire to modify the field and charges on $K_L'\setminus \underline K_L$ are straightforward and we leave them to the interested reader.\\

Section 6 of \cite{ps} is devoted to the proof of Proposition 6.1 there, which is the precise analogue of Proposition~\ref{screening} here. The main difference to \cite{ps} is that here we work with nonconstant $\rho$. The changes to be made are as follows. Here $\bar K_L, K_L'$ correspond to $K_R, K_R'$ of \cite[Prop. 6.1]{ps}, respectively.\\

In our case we work in the simplified situation where points in $\Lambda$ have a minimal separation of $r_0>0$ and multiplicity $1$ for all $p\in\Lambda$. Although most estimates work for more general $\Lambda$ as considered in \cite{ps} too, our present restrictions allow to formulate the conclusion in terms of $\mathcal W_\eta$ rather than in terms of the quantity $\int\yg|E_\eta|^2$ used in \cite{ps}.\\

\textbf{Part 1.} \textit{Changes in the lemmas of \cite[Sec. 6]{ps}:}
\begin{itemize}
\item In all of \cite[Sec. 6]{ps}, we replace at all instances the volume measures $|A|$ of subsets (usually hyperrectangles) $A\subset \mathbb R^d$, by the integrals $\int_A\rho$. For example the hypothesis $|K_R|\in\mathbb N$ of Proposition 6.1 is replaced by our hypothesis $\int_{\bar K_L}\rho\in\mathbb N$ here, which plays the same role. 
\item The vertical part of our domains is now endowed with a separate scale $t$. Thus $t$ plays the role that in \cite{ps} was that of $\ep^2R$, throughout the whole proof.

\item \cite[Lem. 6.3]{ps} has to be replaced by our Lemma~\ref{subdivision} here.
\item In \cite[Lem. 6.5]{ps} we replace the constant $m$ by a function $\rho$ such that $0<\underline\rho\le\rho(x)\le\overline\rho$ throughout $\mathcal R$. The constant $C$ in \cite[Lem. 6.5]{ps} now depends on $\underline\rho,\overline\rho$ as well.
\item \cite[Lem. 6.6]{ps} remains unchanged.
\end{itemize}
\textbf{Part 2.} \textit{Changes in the proof of \cite[Prop. 6.1]{ps}}:\\
We apply the following adaptations, besides the replacement of $d$-dimensional volumes by $\rho$-masses everywhere.
\begin{itemize}
\item The hypotheses analogous to \eqref{enbd1}, \eqref{decv} of our Proposition~\ref{screening} are the same as the estimates of step 1 of \cite[Prop. 6.1]{ps} except for the different choices of constants. 
\item The constant $C_0$ is defined only in case $k=1$. It is now defined using the width $t'\sim t$ given by the thesis of Proposition~\ref{goodbdry} rather than $\ell$ like in \cite{ps}. The role of the constant $C_0$ is to compensate the boundary datum $E_\eta\cdot \nu$ along $\partial K_L'\times[-t',t']$ by a constant boundary datum over $K_L'\times\{-t',t'\}$, thus we define:
\begin{equation*}
C_0:=(2(t')^\gamma|\bar K_L\setminus K_L'|)^{-1}\int_{(\bar K_L\setminus K_L')\times[-t',t']}\yg E_\eta\cdot \nu.
\end{equation*}
\item We tile a neighborhood of $\partial K_L'$ in $\bar K_L\setminus K_L'$ by hyperrectangles $H_i$ of size $\sim t$ and we define again $n_i$ like in \cite{ps} as the total mass of the smeared charges that intersect $\partial K_L'$ restricted to $H_i$.
\item For defining the hyperrectangles $H_i$ we first use a subdivision into size $\sim t$ strips of the $t$-neighborhood above, done along the lines of Proposition~\ref{subdivision}. Observe that measure $\rho + \sum_p\delta_p^{(\eta)}$ is larger than $\underline \rho$ and $x\mapsto \int_{\{x\}\times [-t,t]^k}\yg E_\eta\cdot \nu$ is by Cauchy-Schwartz inequality an $L^2_{loc}$ function on $\partial K_L'$, therefore up to perturbing the $H_i$ slightly (following which also the $n_i$ will change with continuity) we obtain
\[
\int_{H_i} \rho - c_{s,d}^{-1}\left(\int_{\partial D_0\cap \partial \tilde H_i}\yg E_\eta\cdot\nu + 2C_0 t^\gamma|H_i| + n_i\right)\in \mathbb N.
\]
\item Then we define the negative constants $-\bar m_i$ which play the same role as $(m_i-1)$ in \cite{ps}, as follows
\begin{equation}\label{barmi}
c_{d,s}\bar m_i|H_i| = \int_{\partial D_0\cap \partial \tilde H_i}\yg E_\eta\cdot\nu - 2C_0 t^\gamma|H_i| - n_i.
\end{equation} 
This implies that 
\[
\int_{H_i}(\rho - \bar m_i)\in \mathbb N.
\]
\item The requirement replacing $|m_i - 1|<1/2$ in \cite{ps} is that $|\bar m_i|\le \underline \rho/2$. Assuming this to hold, we then have $\rho - \bar m_i\in [\underline\rho/2, \overline\rho +\underline \rho/2]$.
\item The $\mathcal R_\alpha$ are produced via our new Lemma~\ref{subdivision}, applied to the density $\rho - \bar m_i$, which by the two previous points satisfies the hypotheses in that lemma.
\item The bounds on sidelengths of the $\mathcal R_\alpha$ described in \cite[par. following (6.43)]{ps} now contain another bounded factor $\overline\rho^{d-1}/\underline\rho^d$. 
\item The scale of the hyperrectangles in \cite{ps} was $\ell = \ep^2R$, whereas here we take the $H_i$ of size $\sim t$, as given in the good boundary Proposition~\ref{goodbdry}. 
\item We find an analogue of \cite[(6.45)]{ps} in our setting, by multiplying all the terms by a constant factor dependent on $\underline \rho, \overline\rho$, because $\int_{H_i}\rho\ge \underline\rho|H_i|$ and the factors $t^{-d}$ in the right hand side are obtained from the analogue of \cite[(6.34)]{ps}(we recall again that the parameter $t$ here corresponds to $\ell$ of \cite{ps} in this case) via the comparison $ct^d\le|H_i|\le Ct^d$, following from \cite[Lem. 6.3]{ps} in \cite{ps}, whereas in our setting, Lemma~\ref{subdivision} gives the extra factors as above.
\item Step 4 of the proof of \cite[Prop. 6.1]{ps} provides the conditions ensuring $|m_i-1|\le 1/2$. Our $-\bar m_i$ are defined precisely like $m_i-1$ in that proof, thus our corresponding desired bound $|\bar m_i|\le \underline \rho/2$ can be proved in the same way. The bounds we obtain change only by a constant depending on $\underline\rho, \overline\rho, d$. One simplification is that the sum of $\bar n_\alpha,\ \alpha\in I_i$ is in our case bounded by $C r_0^dt^{d-1}$ using the charge separation result of Proposition~\ref{separation}, therefore the corresponding term, which appears in our analogue of \cite[(6.45)]{ps} multiplied by a $t^{-d}$ factor, is bounded for $t$ large. The remaining estimates give precisely the condition \eqref{condLt}.
\item The construction of the four vector fields $E_{i,1},\ldots, E_{i,4}$ is conducted with the following modifications:
\begin{itemize}
\item The fields $E_{i,1}, E_{i,2}$ are defined exactly as in \cite{ps}.
\item The field $E_{i,3}$ has as the only change the replacement of the constant charge $1-m_i$ by the constant $\bar m_i$.
\item The field $E_{i,4}$ is still defined as a superposition of fields obtained like in \cite[Lem. 6.5]{ps} on the $\mathcal R_\alpha$ that cover $H_i$, but now we use the background negative charge $-(\rho -\bar m_i)$, to which we add a charge at the center of $\mathcal R_\alpha$. Note that the estimate of \cite[Lem. 6.5]{ps} depends only on elliptic estimates and holds for the new charges that we use here. The constant obtained in such estimate now depends on $\underline\rho, \overline\rho$.
\end{itemize}
\end{itemize}

\end{proof}

\subsection{Choice of parameters}\label{choiceparsec}

In this subsection we are going to establish a set of parameters, which ensure that \eqref{condLt} holds and that \eqref{energybigger}, \eqref{energybigger2} are $\ll 1$. In particular, this will imply a condition on $\varepsilon_1$, which says that a good bound on $err_{sc}$ can be expected only for $\ve_1$ not too small. 

\subsubsection{Case $k=0$}
In the case $k=0$, one can write $\varepsilon_1$ as a power of $T$ and choose $l$ and $t$ to be powers of $L$. More precisely, let $\varepsilon_1$, $l$ and $t$ be such that 
\begin{align*}
\varepsilon_1=T^{\left(\frac{1-\delta}{\delta}\right)d},\ l=L^{b}\ \text{and}\ t=L^{\theta}
\end{align*}
with $\delta>1$, $b<1$ and $\theta<1$. Recall that $L$ is such that $\varepsilon_1^dT=T^{1/\delta}\le L\le T$ which gives $T\le L^{\delta}$.
A straightforward calculation shows that  \eqref{condLt} holds if
$$
\theta\ge\frac{\delta d-b}{d+1}
$$
and 
\begin{equation*}
err_{sc}(\eta,t,L,l,\ve_1,C_1, 0)=C_1g(\eta)L^{(\delta-1)d-b}+C_1L^{\theta+(\delta-1)-b}+g(\eta)L^{\theta-1}.
\end{equation*}
As a consequence to ensure that $err_{sc}(\eta,t,L,l,\ve_1,C_1, 0)\ll 1$ we have to choose
$$
\theta<b+(1-\delta)d.
$$
Note that, since $\delta>1$, $b+(1-\delta)d<b<1$. To sum up, $\theta$ as to be chosen such that
\begin{equation}\label{condtheta}
\frac{\delta d-b}{d+1}\le \theta< b+(1-\delta)d.
\end{equation}
This is possible if and only if $\frac{\delta d-b}{d+1}< b+(1-\delta)d$. As a consequence $\delta$ has to be such that 
\begin{equation}\label{conddelta}
1<\delta<1+\frac{b(d+2)-d}{d(d+2)}
\end{equation}
which in particular gives a condition on the scale $\varepsilon_1T$. Such a $\delta$ exists if and only if $\frac{b(d+2)-d}{d(d+2)}>0$. Hence $b$ has to be chosen such that
\begin{equation}\label{condb}
b>\frac{d}{d+2}.
\end{equation}
Finally remark that, with this choice of parameters $(\delta-1)d-b<0$ and 
\begin{equation}\label{smallerrsckzero}
err_{sc}(\eta,t,L,l,\ve_1,C_1, 0)=(g(\eta)+1)(1+C_1)o_{L\to \infty}(1).
\end{equation}
\subsubsection{Case $k=1$} In the case $k=1$, the situation is more delicate since we do not know explicitly the decay rate of $C_2(E,t,L)$. Nevertheless, in what follows, we need to apply Proposition~\ref{screening} to a sequence of fields $E_n'$ such that $$\lim_{t\to+\infty}\lim_{L\to\infty}\lim_{n\to\infty} C_2(E_n',t,L)=0.$$ Hence, for all $\varepsilon_2>0$, for $L\gg \ve_2^{1/2}$ and $L,n$ large enough there holds
\begin{equation*}
       \frac{1}{L^d}\int_{K_L\times\left(\mathbb R\setminus[-\frac12 \ve_2^{1/2} L, \frac12 \ve_2^{1/2} L]\right)}|y|^\gamma|E_{n,\eta}'|^2\le \varepsilon_2.
\end{equation*}
Hence, by choosing $t=\ve_2^{1/2}L, l=\ve_2^{1/4}L$ and $C_2=\ve_2$ in Proposition~\ref{screening}, we have that the first condition in~\eqref{condLt} holds. The second one is satisfied if  
$$
C_1 \ve_2^{-\frac{2(d-\gamma)+1}{4}} L^{(\gamma-1)}\ve_1^{-1}\le c_0^2
$$
that is if
$$
L\ge c_0^{-\frac{2}{1-\gamma}}\ve_1^{-\frac{1}{1-\gamma}} \ve_2^{-\frac{2(d-\gamma)+1}{4(1-\gamma)}}.
$$
Moreover, 
\begin{equation*}
err_{sc}(\eta, \ve_2^{1/2} L,L, \ve_2^{1/4} L,\ve_1,C_1,  \ve_2)=C_1\frac{1}{\ve_1\ve_2^{1/4} L}g(\eta)+\ve_2^{1/2} +C_1\frac{\ve_2^{1/4}}{\ve_1}  +g(\eta)\ve_2^{1/2}.
\end{equation*}
To ensure that $err_{sc}(\eta, \ve_2^{1/2} L,L, \ve_2^{1/4} L,\ve_1,C_1,  \ve_2)\ll 1$, $\ve_1$ has to be such that 
\begin{equation}\label{condeps1}
1\ge \ve_1 \gg \ve_2^{\frac14}.
\end{equation}
In this case the above bounds on $L$ are implied by the condition
\begin{equation}\label{condLkone}
L\ge c_0^{-\frac{2}{1-\gamma}} \ve_2^{-\frac{(d-\gamma)+1}{2(1-\gamma)}}.
\end{equation}
To summarize, if \eqref{condeps1}, \eqref{condLkone} hold then
\begin{equation}\label{smallerrsckone}
err_{sc}(\eta, \ve_2^{1/2} L,L, \ve_2^{1/4} L,\ve_1,C_1,  \ve_2)=(1+g(\eta))(1+C_1)o_{L,n\to \infty}(1).
\end{equation}

\section{Proof of~\eqref{mainlimit} of Theorem~\ref{mainthm}}\label{proofnocrenel}
By using the fact that for a minimizer of $H_n$ all the point are in $\Sigma$, and in view of Proposition~\ref{splitting} and the result of \cite[Thm. 4]{ps}, we have the a priori bound
\begin{equation}\label{boundWeta}
\mathcal W_\eta(E'_n,\R^d)\le n \int_{\Sigma}\min_{\mathcal{A}_{m'_V(x)}} \mathcal W\,dx+o_{n\to+\infty}(n)+o_{\eta\to 0}(1).
\end{equation}
Moreover, in case $k=1$ we work under the assumption \eqref{hypdecayfields}. Finally, we will use that since $m_V$ is $C^{0,\alpha}$ in $\Sigma$, we have
\begin{equation}\label{eqgradmV}
\| m_V'\|_{C^{0,\alpha}(\Sigma')} \le \frac{\|m_V\|_{C^{0,\alpha}(\Sigma)}}{n^{\alpha/d}},
\end{equation}
and, whenever $K_T(a)\subset \Sigma'$,  $T\lesssim{n^{1/d}}$, $0\le \beta<\alpha\le 1$ there holds
\begin{equation}\label{petitesvar}
\| m'_V\|_{\mathcal C^{0,\alpha}(K_T(a))}T^{\beta}\le C\frac{1}{({n^{1/d}})^{\alpha-\beta}}\le o_{n\to+\infty}(1)
\end{equation}
The proof is based on a bootstrap argument as in \cite{rns} : by a mean value argument, using the a priori bound on the energy~\eqref{boundWeta} and $T=n^{1/d}$ as initial scale, we can find a square close to $K_\ell(a)$ which has a good boundary, \emph{i.e.} such that~\eqref{enbd1}, \eqref{decv} are satisfied (relative to $\ell$). This is only possible if $\ell$ is not too small compared to $n^{1/d}$, more precisely if $\ell_i\in [\ve_1^{1/d}n^{1/d}, n^{1/d}]$ for $i=1,\ldots,d$. If indeed $\ell$ is such that $\ell_i\in [\ve_1^{1/d}n^{1/d}, n^{1/d}]$ then we are essentially done: a comparison argument in the hypercube with the good boundary allows to conclude. More precisely, we use the following result. We note here that in this section we omit the index $n$ on $a$ in order to lighten up the notation.

\begin{proposition}\label{compgbdry} Let $C_1$ a positive constant. Let $t,L,l,\ve_1,\ve_2$ and $E'_n$ be as in Section~\ref{choiceparsec}. Assume that bounds \eqref{enbd1} and \eqref{decv} hold for $E'_n$ in $K'_L(a)\subset K_L(a)$ with $\op{dist}(\partial K_L'(a), \partial K_L(a))\in[l, 2l[$. Then 
\begin{equation}\label{compgbdryest}
\left|\frac{\mathcal W_\eta(E'_n,K'_L(a))}{|K'_L(a)|}-\frac{1}{| K_L(a)|}\int_{K_L(a)} \min_{\mathcal A_{m_V'(x)} }\mathcal W dx\right|\le (1+g(\eta))(1+C_1)o_{L,n\to\infty}(1)+o_{\eta\to0}(1).
\end{equation}
\end{proposition}

Next, if $\ell_i$ is smaller than $\ve_1^{1/d}n^{1/d}$ for some $i$, we bootstrap the argument: we first obtain by the above argument a control of the energy and the number of points on a hypercube of size $\ve_1^{1/d}n^{1/d}$ containing $K_\ell(a)$, and then we re-apply the reasoning starting from that hypercube. This allows to go down to a smaller scale, and we iterate the procedure until we reach the desired value of $\ell$. This iteration will not cumulate error, its only main restriction is that the final square will have to be at a certain distance away from $\partial \Sigma'$, because of the repeated mean value arguments.

More precisely, we proceed as follows for the proof of~\eqref{mainlimit} of Theorem~\ref{mainthm}. Let $K_\ell(a)$ as in the statement of Theorem~\ref{mainthm}. We set
\begin{equation}\label{defC1}
C_1=\max_{m\in [\underline m,\overline m]}\min_{\mathcal A_m}\mathcal W+C_{s,d,\overline m}g(\eta)+1
\end{equation}
and
\begin{equation}\label{deflell}
l_{\ell}:=\left\{
\begin{aligned} &\ell^b & \text{if } k=0\\
&\varepsilon_2^{1/4} \ell & \text{if } k=1
\end{aligned}
\right.
\end{equation}
with $b<1$ and $\varepsilon_2$ chosen as in Section~\ref{choiceparsec}. 
Without loss of generality we may assume $\op{dist}( K_\ell(a), \partial\Sigma')\ge 3l_{\ell}$. In this case $\op{dist}( K_\ell(a), \partial\Sigma')\ge \max(d_n, 3l_{\ell})$ where 
$$
d_n:=\left\{
\begin{aligned} &n^{q/d} & \text{if } k=0\\
&\varepsilon_0 n^{1/d} & \text{if } k=1
\end{aligned}
\right..
$$
The proof for the case where $d_n\le \op{dist}(K_\ell(a), \partial \Sigma')< 3l_{\ell}$ would start by subdividing the cube $K_\ell(a)$ into cubes of the smallest size $L_{min}$ permitted by the choices of Section~\ref{choiceparsec}, namely $L_{min}\sim 1$ for $k=0$ and $L_{min}\sim \ve_2^{-\frac{2(d-\gamma)+1}{2(1-\gamma)}}$ for $k=1$. The smaller cubes have $3l_{L_{min}}\le d_n$ respectively if $n$ is large enough in case $k=0$ and if $\ve_2$ is small enough (which by Section~\ref{choiceparsec} means that $\ell,n$ must be large enough depending on \eqref{hypdecayfields}) in case $k=1$. Therefore the result for $\op{dist}( K_{L_{min}}(a), \partial\Sigma')\ge 3l_{L_{min}}$ holds for all cubes in the subdivision, and summing the bounds \eqref{mainlimit} for all these cubes we obtain \eqref{mainlimit} for $K_\ell(a)$ too.

We split the proof of Theorem~\ref{mainthm} into two cases.

{\bf Case $1$: $\ell+3l_{\ell}\ge \ve_1^{1/d}n^{1/d}$.} Let us then define the scale $L^{(1)}=\ell+3l_{\ell}$. Since we assumed that $\op{dist}( K_\ell(a), \partial\Sigma')\ge 3l_{\ell}$, we have $K_{\ell+3\l_{\ell}}(a)\subset \Sigma'$, and so there exists a center $a^{(1)}$ such that
$$
K_{\ell+3\l_{\ell}}(a)\subset K_{L^{(1)}}(a^{(1)})\subset \Sigma'.
$$ 
Then we apply Proposition~\ref{goodbdry} with $K_T$ replaced by $\Sigma'$ itself, $C_1$ given by \eqref{defC1} and $L=L^{(1)}$. Moreover $t,C_2,  \varepsilon_1$ and $l$ are chosen as in Section~\ref{choiceparsec}. This gives the existence of a square $K'_{L^{(1)}}(a^{(1)})\subset K_{L^{(1)}}(a^{(1)})$ such that $\op{dist}(\partial K'_{L^{(1)}}(a^{(1)}), \partial K_{L^{(1)}}(a^{(1)}))\in[l, 2l[$ and \eqref{enbd1} and \eqref{decv} are satisfied. Note that for $\ell$ large enough 
$$
K_{\ell}(a)\subset K'_{L^{(1)}}(a^{(1)})\subset K_{L^{(1)}}(a^{(1)})\subset \Sigma'.
$$
Then by applying Proposition~\ref{compgbdry}, we deduce
\begin{equation}\label{estimationproofth}
\left|\frac{\mathcal W_\eta(E'_n,K'_{L^{(1)}}(a^{(1)}))}{|K'_{L^{(1)}}|}-\frac{1}{| K_{L^{(1)}}|}\int_{K_{L^{(1)}}(a^{(1)})} \min_{\mathcal A_{m_V'(x)} }\mathcal W dx\right|\le (1+g(\eta))(1+C_1)o_{\ell,n\to\infty}(1)+o_{\eta\to0}(1).
\end{equation}

By using charges separation and the fact that $K_{\ell}(a)\subset K'_{L^{(1)}}(a^{(1)})$, we obtain the upper bound
\begin{align*}
\frac{\mathcal W_\eta(E'_n,K_{\ell}(a))}{|K_{\ell}|}&\le \frac{1}{|K_{L^{(1)}}|}\int_{K_{L^{(1)}}(a^{(1)})} \min_{\mathcal A_{m_V'(x)} }\mathcal W dx+g(\eta)^2o_{\ell,n\to\infty}(1)+o_{\eta\to0}(1)\\
&=\frac{1}{|K_{\ell}|}\int_{K_{\ell}(a)} \min_{\mathcal A_{m_V'(x)} }\mathcal W dx+g(\eta)^2o_{\ell,n\to\infty}(1)+o_{\eta\to0}(1).
\end{align*}
To obtain the lower bound, we may apply once again Proposition~\ref{goodbdry} at the scale $\ell$ to obtain a square $K'_\ell(a)\subset K_\ell(a)$ with a good boundary and such that 
$$
\frac{\mathcal W_\eta(E'_n,K_{\ell}(a))}{|K_{\ell}|}\ge \frac{\mathcal W_\eta(E'_n,K'_{\ell}(a))}{|K'_{\ell}|}-g(\eta)o_{\ell\to\infty}(1).
$$
Hence the application of Proposition~\ref{compgbdry} to the square $K'_\ell(a)$ leads to the desired result.

{\bf Case $2$: $\ell+3l_{\ell}< \ve_1^{1/d}n^{1/d}$.} Let $L^{(1)}=\ve_1^{\frac1d}n^{\frac1d}$. Since we have $\op{dist}( K_\ell(a), \partial\Sigma')\ge d_n$, we have $K_{\ell+d_n}(a)\subset \Sigma'$, and so there exists a center $a^{(1)}$ such that
$$
K_{\ell+d_n}(a)\subset K_{L^{(1)}}(a^{(1)})\subset \Sigma'.
$$ 
Then we apply Proposition~\ref{goodbdry} with $K_T$ replaced by $\Sigma'$ itself, $C_1$ given by \eqref{defC1} and $L=L^{(1)}$. Moreover $t,C_2,  \varepsilon_1^{(1)}$ and $l^{(1)}$ are chosen as in Section~\ref{choiceparsec}. This gives the existence of a square $K'_{L^{(1)}}(a^{(1)})\subset K_{L^{(1)}}(a^{(1)})$ such that $\op{dist}(\partial K'_{L^{(1)}}(a^{(1)}), \partial K_{L^{(1)}}(a^{(1)}))\in[l^{(1)}, 2l^{(1)}[$ and \eqref{enbd1} and \eqref{decv} are satisfied. Note that for $\ell$ large enough there holds
$$
K_{\ell+d_n-2l^{(1)}}(a)\subset K'_{L^{(1)}}(a^{(1)})\subset K_{L^{(1)}}(a^{(1)})\subset \Sigma'.
$$
Then by applying Proposition~\ref{compgbdry}, we have~\eqref{estimationproofth} as before. Next, if $\ell$ is large enough, the hypotheses of Proposition~\ref{goodbdry} are satisfied in $K'_{L^{(1)}}$ (which plays the role of $K_T$ in Proposition~\ref{goodbdry}), and with the same constant $C_1$. This a consequence of~\eqref{estimationproofth} for $\ell, n$ and $1/\eta$ large enough. We can thus re-apply Proposition~\ref{goodbdry} in $K'_{L^{(1)}}$
and with new subscale $L^{(2)}=\max(\ell+3l_\ell,(\ve_1^{(1)})^{1/d}L^{(1)})$. If $\ell+3l_\ell\ge (\ve_1^{(1)})^{1/d}L^{(1)}$, we conclude as in the Case $1$. If $\ell+3l_\ell< (\ve_1^{(1)})^{1/d}L^{(1)}$, we iterate the above procedure. In the end, we obtain a sequence of $L^{(j)}$ with $L^{(j)}=\max(\ell+3l_\ell,(\ve_1^{(j-1)})^{1/d}L^{(j-1)})$ and such that  
$$
K_{\ell+d_n-2\sum_{j=1}^{\bar j}l^{(j)}}(a)\subset K'_{L^{(\bar j)}}(a^{(\bar j)})\subset K_{L^{(\bar j)}}(a^{(\bar j)})\subset\ldots\subset K'_{L^{(1)}}(a^{(1)})\subset K_{L^{(1)}}(a^{(1)})\subset \Sigma'.
$$
where $\bar j$ is the smallest integer such that $\ell> \prod_{j=2}^{\bar j}(\ve_1^{(j-1)})^{1/d}(\ve_1 n)^{1/d}=L^{(\bar j)}$. Now, we have to ensure that $K_\ell(a)$ is a subset $K'_{L^{(\bar j)}}(a^{(\bar j)})$. This is true if $d_n-2\sum_{j=1}^{\bar j}l^{(j)}>0$. We proceed as follows.
\begin{itemize}
\item {\bf Case $k=0$.} In this case, by Section~\ref{choiceparsec}, it follows that $L^{(j)}=(L^{(j-1)})^{1/\delta}=(n^{1/d})^{\delta^{-j}}$ and $l^{(j)}=(L^{(j)})^b$. Bounding each time $L^{(j)}$ by $n^{1/(\delta d)}$ we have that $\sum_{j=1}^{\bar j}l^{(j)}\le  \bar j n^{b/(\delta d)}$. Moreover, $\bar j$ is finite. Indeed, $\bar j$ is defined as the smallest integer such that $\ell> \prod_{j=2}^{\bar j}(\ve_1^{(j-1)})^{1/d}(\ve_1 n)^{1/d}=(n^{1/d})^{\delta^{-\bar j}}$ which gives 
$$
\bar j= \left[\frac{\log\frac{\log n^{1/d}}{\log \ell}}{\log\delta}\right]. 
$$
Hence, to have  
$$
d_n-2\sum_{j=1}^{\bar j}l^{(j)}\ge n^{q/d}-2 \bar j n^{b/(\delta d)}>0
$$
it is enough to choose $1>q>\frac{b}{\delta}$.
\item {\bf Case $k=1$.}  By the $k=1$ cases in Section~\ref{choiceparsec}, it now follows that $l^{(j)}=\ve_2^{1/4}L^{(j)}$
and $L^{(j)}=\prod_{r=2}^{j}(\ve_1^{(r-1)})^{1/d}(\ve_1 n)^{1/d}$. So to deal with this case, we fix $\ve_1^{(r)}=\ve_2^{1/8}$ for all $r\in [\![ 1,\bar j ]\!]$. This is compatible with condition on $\ve_1$ in Section~\ref{choiceparsec}. As a consequence $L^{(j)}=\ve_2^{j/(8d)}n^{1/d}$ for all $j\in [\![ 1,\bar j ]\!]$ and 
$$
\sum_{j=1}^{\bar j}l^{(j)}=  \ve_2^{1/4} n^{1/d} \sum_{j=1}^{\bar j} (\ve_2^{1/(8d)})^j\le \ve_2^{1/4}\frac{\ve_2^{1/(8d)}}{1-\ve_2^{1/(8d)}}n^{1/d}.
$$
for all $\ve_2>0$ provided that $\ell$ and $n$ are large enough. Therefore if $d_n=\ve_0 n^{1/d}$ with $\ve_0>0$, then $d_n-2\sum_{j=1}^{\bar j}l^{(j)}>0$.
\end{itemize}

\subsection{Comparison argument on hypercubes with good boundary}

This subsection is devoted to the proof of Proposition~\ref{compgbdry}. We start with the following lemma.

\begin{lemma}\label{lemminenergy}
Let $(x_1,\cdots,x_n)$ be a minimizer of $H_n$. Let $\mu_V=m_V(x)dx,\ \mu_V'=m_V'(x')dx'$ be respectively the equilibrium measure and its blow-up at scale $n^{1/d}$ as above and let $E_n'=\nabla h_n'$ be a sequence of blown-up vector fields corresponding to these minimizers as in \eqref{hnp}. Let $\Sigma'$ be the support of $\mu_V'$. Then for all $\Omega \subset \Sigma'$, we have
\begin{equation}\label{minoveromega}
	\mathcal W_\eta(E'_n,\Omega)\le \mathcal W_\eta(E,\Omega)+o_{\eta\to 0}(1)
\end{equation}
for any $E=\nabla h$ which satisfies  
\begin{equation*}
\left\{
\begin{aligned}
&-\op{div} (\yg \nabla h) = c_{s,d} \Big( \sum_{i=1}^n \delta_{p'_i}-m'_V\delta_{\mathbb R^d} \Big) \quad& \text in\ \Omega\times\R^k\\
&E\cdot \vec{\nu}=E'_n\cdot \vec{\nu} \quad& \text on \ \partial \Omega \times \R^k
\end{aligned}
\right.
\end{equation*}
with all the points $p_i'$ in $\Sigma'$. 
\end{lemma}
The proof of the lemma has some similarity to \cite[Sec. 5.2]{rns}.
\begin{proof} Let $(x_1,\cdots,x_n)$ minimize $H_n$ and $E_n'=\nabla h_n'$ be as in \eqref{hnp}. Since all points are in $\Sigma$ by Proposition~\ref{separation}, in view of Proposition~\ref{splitting}, by minimality of $(x_1,\cdots,x_n)$, we have
\begin{equation}\label{comparebareen}
	\lim_{\eta \to 0}\mathcal W_\eta(E'_n,\R^d)\le \lim_{\eta\to 0}\mathcal W_\eta(\bar E,\R^d)
\end{equation}
for any $\bar E=\nabla h$ which satisfies  
$-\op{div} (\yg \nabla h) = c_{s,d} \Big( \sum_{i=1}^n \delta_{p'_i}-m'_V\delta_{\mathbb R^d} \Big)$ in $\R^{d+k}$
with all the points $p_i'$ in $\Sigma'$. Note that we may extend $E$ as in the statement (provided $\int_\Omega \mu'_V$ is such that such an $E$ exists) by $E_n'$ outside $\Omega$ in order to obtain a competitor $\bar E$ as above. Moreover as a consequence the energy contributions to \eqref{comparebareen} of $\bar E, E'_n$ outside $\Omega$ coincide. This proves our result.
\end{proof}

As a consequence of Lemma~\ref{lemminenergy}, we have that, for all $\Omega\in \Sigma'$,
\begin{equation}
 \frac{\mathcal W_\eta(E'_n,\Omega)}{|\Omega|}=\min_{E\in \mathcal{A}_{m'_V},\ E\cdot\vec{\nu}=E'_n\cdot\vec{\nu} } 
\frac{\mathcal W_\eta(E,\Omega)}{|\Omega|} + o_{\eta\to 0}(1).
\end{equation}
We now use parameters as in Section~\ref{choiceparsec}

Next, we remark that the vector field $E'_n$ satisfies the hypotheses of Proposition~\ref{screening}. Hence, there exists $\underline K_L(a)\subset K'_L(a)\subset K_L(a)$ such that by \eqref{smallerrsckzero} or \eqref{smallerrsckone} 
\begin{align}\label{ubmaincube}
\frac{\mathcal W_\eta(E'_n,K'_L(a))}{|K'_L(a)|}\le \sigma_{0}(\underline K_L(a);m'_V)+\left(1+g(\eta)\right)o_{L\to \infty}(1) + o_{\eta\to 0}(1),
\end{align}
with $\sigma_{0}(K;m'_V)=\min_{E\in \mathcal{A}_{m'_V},\ E\cdot\vec{\nu}=0 } \frac{\mathcal W_\eta(E,K)}{|K|}$. To prove the claim \eqref{ubmaincube} we apply Proposition~\ref{screening}, which allows to construct a vector field $\tilde E$ and a subset $\tilde \Lambda\subset K_L(a)$ such that 
\[
\left\{
\begin{aligned}
&-\op{div}(\yg \tilde E)=c_{d,s}\left(\sum_{p\in\tilde \Lambda}\delta_p - m'_V(x)\delta_{\mathbb R^d}\right) & \text{ in }K_{L}(a)\times\R^k\\
&\tilde E\cdot\vec{\nu}=0  & \text{ on }\partial\underline K_{L}(a)\times\R^k\\
&\tilde E\cdot\vec{\nu}=E'_n\cdot \vec{\nu}  & \text{ on }\partial K'_{L}(a)\times [-t,t]^k\\
&\tilde E=E'_n , \ \tilde \Lambda=\Lambda & \text{ in } (K_L(a)\setminus K'_{L}(a))\times [-t,t]^k\\
\end{aligned}
\right..
\]
and
\begin{align*}
\frac{1}{L^d} \int_{(K_L'\setminus \underline K_L)\times\mathbb R^k}\yg|\tilde E_\eta|^2 \le C err_{sc}(\eta,t,L,l,\ve_1,C_1, C_2).
\end{align*}
Due to the choices of parameters like in Section~\ref{choiceparsec}, the quantity $err_{sc}$ is bounded by $(1+g(\eta))(1+C_1)o_{L,n\to \infty}(1)$. The bound \eqref{ubmaincube} then follows using the bound on the separation of charges and a packing argument to bound the number of charges in $K_L'\setminus\underline K_L$ by $C|K_L'\setminus\underline K_L|\le CtL^{d-1}$.

Hence, by defining $\tilde E=\argmin_{E\in \mathcal{A}_{m'_V},\ E\cdot\vec{\nu}=0 }\frac{\mathcal W_\eta(E, \underline K_L)}{| \underline K_L|}$ on $ \underline K_L(a)$ and by using Lemma~\ref{lemminenergy}, we obtain 
$$
\frac{\mathcal W_\eta(E'_n, K'_L(a))}{|K'_L(a)|}\le \frac{\mathcal W_\eta(\tilde E, K'_L(a))}{| K'_L(a)|}.
$$  
By the same reasoning, there exists $K'_L(a)\subset \bar K_L(a)$, a subset $\tilde \Lambda \subset \bar K_{L}(a)$ and a vector field $\tilde E'_n$ such that 
\[
\left\{
\begin{aligned}
&-\op{div}(\yg \tilde E'_n)=c_{d,s}\left(\sum_{p\in\tilde \Lambda}\delta_p - m'_V(x)\delta_{\mathbb R^d}\right) & \text{ in }\bar K_{L}(a)\times\R^k\\
&\tilde E'_n\cdot\vec{\nu}=0  & \text{ on }\partial\bar K_{L}(a)\times\R^k\\
&\tilde E'_n\cdot\vec{\nu}=E'_n\cdot \vec{\nu}  & \text{ on }\partial K'_{L}(a)\times [-t,t]^k\\
&\tilde E'_n=E'_n , \ \tilde \Lambda=\Lambda & \text{ in } K'_{L}(a)\times [-t,t]^k\\
\end{aligned}
\right.
\]
and 
\begin{align}\label{lbmaincube}
\frac{\mathcal W_\eta(E'_n,K'_L(a))}{|K'_L(a)|}\ge &\frac{\mathcal W_\eta(\tilde E'_n, \bar K_L(a))}{|\bar K_L(a)|}-err_{sc}(\eta,t,L,l,\ve_1,C_1, C_2(E_n',t,L))\\
&=\frac{\mathcal W_\eta(\tilde E'_n, \bar K_L(a))}{|\bar K_L(a)|}-\left(1+g(\eta)\right)(1+C_1)o_{L,n\to \infty}(1).
\end{align}
The last step required in order to conclude the proof, is to find an upper bound for $\sigma_{0}(\underline K_L(a);m'_V)$ and a lower bound for $\frac{\mathcal W_\eta(\tilde E'_n, \bar K_L(a))}{|\bar K_L(a)|}$. This is done in the two next subsections.
\subsubsection{Upper bound}
To obtain an upper bound for $\sigma_{0}(K;m'_V)=\min\limits_{E\in \mathcal{A}_{m'_V},\ E\cdot\vec{\nu}=0 } \frac{\mathcal W_\eta(E,K)}{|K|}$, we proceed as in \cite{rns}. In particular we have the following proposition.

\begin{proposition}\label{propub} 
Let $a\in \R^d$ and $R>0$ such that $K_R(a)\subset \Sigma'$. Let $\alpha\in]0,1]$ and $\rho$ be a non-negative $\mathcal C^{0,\alpha}(K_R(a))$ function for which there exists $\underline \rho, \overline \rho >0$ such that $\underline \rho\le \rho(x)\le \overline \rho$. Let $K_R(a)$ be such that $\int_{K_R(a)}\rho(x)\,dx \in \N$. Then there exists a constant $C>0$ depending only on $d,s$ such that for each fixed $\beta\in]0,1+\alpha[$ we have 
\begin{align}\label{upperbound}
\sigma_{0}(K_R(a);\rho)\le&\,\frac{1}{|K_R|}\int_{K_R(a)}\min_{\mathcal{A}_{\rho(x)}} \mathcal W\,dx +C (g(\eta) +1) o_{R\to \infty}(1)+o_{\eta\to0}(1)+C R^{2\beta}\|\rho\|^2_{C^{0,\alpha}(K_R(a))}\nonumber\\
&+C R^{\beta}\|\rho\|_{C^{0,\alpha}(K_R(a))}\left[c_{d,s}\overline \rho g(\eta)+1+(g(\eta) + 1)o_{R\to \infty}(1) +o_{\eta\to0}(1)\right]^{1/2}.
\end{align}	
\end{proposition}
\begin{remark}
Note that the rate of convergence of the above bounds $o_{R\to\infty}(1)$ as $R\to\infty$ depends on $m_V, \beta$. As $\beta\to 0$ this rate degenerates, as can be seen in the proof. In the present formulation, choosing $\beta>\alpha$ provides no advantage, however it may be possible by refining the proof to obtain more precise estimates of $o_{R\to\infty}(1)$ valid only for such large values of $\beta$.
\end{remark}
\begin{proof}
The proof of this proposition is similar to \cite[Prop. 4.1]{rns} except that we have to be more careful with error terms which blow up as $\eta\to 0$. 

\noindent {\bf Step 1.}
We define a smaller scale $r=R^\lambda$ with $\lambda\in]0,1[$ for now, and with $\lambda$ to be fixed at the end of the proof, and we construct a collection $\mathcal K$ of rectangles which partition  $K_R(a)$,  whose sidelengths are between $r-O\left(\frac{1}{r}\right)$ and $r+O\left(\frac{1}{r}\right)$, and such that for all $K\in \mathcal K$ we have $\int_K \rho(x)\,dx\in \mathbb{N}$. This is possible for example via the partitioning lemma~\cite[Lem. 6.3]{ps}.\\

\noindent {\bf Step 2.} We denote by $x_K$ the center of each $K$ and $\rho_{K}=\fint_{K}\rho(x)\,dx$ and we consider $E$ a minimizer of $\mathcal W$. Since as in \cite[Sec. 7]{ps} and \cite{rs,ss2d} we may obtain $E$ by screening minimizing configurations on larger and larger cubes, we may assume due to Proposition~\ref{separation} that charges contributing to $E$ are well-separated and have multiplicity one. Moreover, if $k=1$, we note that, by the periodicity of $E$, we must have  
$$
\lim_{t\to+\infty}\lim_{R\to+\infty}\frac{1}{R^d}\int_{K_R\times(\R\smallsetminus(-t,t))}|y|^{\gamma}|E|^2=0.
$$

Hence, using \cite[Prop. 6.1]{ps}, we obtain in each $K$ a vector field $E_K$ satisfying
$$
\left\{\begin{aligned}&-\op{div} (\yg E_K) = c_{s,d} \Big( \sum_{p\in \Lambda_K} \delta_{p}-\rho_K\delta_{\mathbb R^d} \Big) & \text{ in}\ K\times\R^k\\ &E_K\cdot\vec{\nu}=0 &\text{ on}\ \partial  K\times \R^k \\
&E_K=0 &\text{ outside } K\times \R^k\end{aligned}\right.
$$
for some discrete subset $ \Lambda_K\subset  K$, and
\begin{equation}\label{ubestimenergy1}
\frac{\mathcal W_\eta(E_K,{ K})}{|K|}\le \min_{\mathcal A_{\rho_K}}\mathcal W+C(g(\eta)+1)o(1)_{r\to +\infty}+o_{\eta\to0}(1), 
\end{equation}
where the $o(1)$ terms depend on the approximation of $\min \mathcal W$ in \cite{ps} as described above.\\

Then, we have to rectify the weight $\rho_K$. For $K\in \mathcal K$, we let $h_K$ solve 
\begin{equation*}
\left\{
\begin{aligned}
	&-\op{div} (\yg \nabla h_K) = c_{s,d} (\rho_K- \rho(x))\delta_{\mathbb R^d}\quad & \text{ in}\ K\times [-r, r]^k\\
	&\partial_\nu h_K=0\quad & \text{ on the rest of }\ \partial (K\times [-r, r]^k)\\
\end{aligned}
\right.
\end{equation*}\
and we put $h_K=0$ in $K\times (\R^k\smallsetminus [-1, 1]^k)$.  
As a consequence of elliptic estimates as in the proof of \cite[Lem. 6.4]{ps}, we have 
\begin{align}
\label{eqestimweightk0}
\int_{K} |\nabla h_K|^2\le C r^2\|\rho-\rho_K\|^2_{L^{\infty}(K)}|K|\le C r^{2\alpha+2}\|\rho\|^2_{C^{0,\alpha}(K)}|K| & \quad\text{for } k=0;\\
\label{eqestimweightk1}
\int_{K\times \R^k} \yg |\nabla h_K|^2\le C r^{1-\gamma} \|\rho-\rho_K\|^2_{L^{\infty}(K)}|K|\le C r^{2\alpha+1-\gamma}\|\rho\|^2_{C^{0,\alpha}(K)}|K| & \quad\text{for } k=1.
\end{align}
We then define $\tilde E$ to be $E_K + \nabla h_K$ in each $K\in \mathcal{K}$. Pasting these together defines a  $\tilde E$ over the whole $K_R(a)$, satisfying
$$
\left\{\begin{aligned}&-\op{div} (\yg \tilde E) = c_{s,d} \Big( \sum_{p\in \Lambda} \delta_{p}-\rho(x)\delta_{\mathbb R^d} \Big) \quad& \text{ in}\ K_R(a)\times\R^k\\ &\tilde E\cdot\vec{\nu}=0 &\text{ on}\ \partial  K_R(a)\times \R^k \end{aligned}\right.
$$
for some discrete set $\Lambda$. Next we have to evaluate $\mathcal W_\eta(\tilde E, K_R(a))$. In each $K\in \mathcal K$, we apply the Cauchy-Schwartz inequality to the squared  $L^2_\gamma$-norm of $\tilde E$, followed by \eqref{ubestimenergy1} and estimating the $\nabla h_K$-term via the weaker bound \eqref{eqestimweightk0} in both cases $k=0,1$, rather than using the more precise \eqref{eqestimweightk1} for $k=1$. This gives:
\begin{align*}
\mathcal W_\eta(\tilde E,K) \le &\,\mathcal W_\eta(E_K,K)+\int_{K\times \R^k} \yg |\nabla h_K|^2
+2\left(\int_{K\times \R^k} \yg |E_K|^2\right)^{1/2}\left(\int_{K\times \R^k} \yg |\nabla h_K|^2\right)^{1/2}\\
&\le|K|\min_{\mathcal{A}_{\rho_K}} \mathcal W +C|K|(g(\eta) + 1)o_{r\to \infty}(1)+|K|o_{\eta\to0}(1)+C |K|r^{2\alpha+2}\|\rho\|^2_{C^{0,\alpha}(K)}\\
&+C|K|r^{\alpha+1}\|\rho\|_{C^{0,\alpha}(K)}\left(\frac{\mathcal W_\eta(E_K,K)}{|K|}+c_{d,s}\rho_K g(\eta)\right)^{1/2}\\
&\le|K|\min_{\mathcal{A}_{\rho_K}} \mathcal W +C|K|(g(\eta) + 1)o_{r\to \infty}(1)+|K|o_{\eta\to0}(1)+C |K|r^{2\alpha+2}\|\rho\|^2_{C^{0,\alpha}(K)}\\
&+C|K|r^{\alpha+1}\|\rho\|_{C^{0,\alpha}(K)}\left[1+(g(\eta) + 1)o_{r\to \infty}(1)+o_{\eta\to0}(1)+c_{d,s}\rho_K g(\eta)\right]^{1/2}
\end{align*}
We now sum over all hypercubes $K$ as above and we use the subadditivity of $A\mapsto \mathcal W_\eta(E,A)$ in order to add the contributions of the above left-hand sides to bound $\mathcal W_\eta(\tilde E,K_R(a))$. On the right hand side all terms except the first give contribution of $|K_R|$ multiplied by the error terms. For the first term, we proceed using the scaling $\min_{\mathcal A_{\rho}}\mathcal W= \rho^{1+s/d}\min_{\mathcal A_1}\mathcal W$ for power kernels and $\min_{\mathcal A_{\rho}}\mathcal W= \rho\min_{\mathcal A_1}\mathcal W - \tfrac{c_{0,d}}{d} \rho\log\rho$ for logarithmic ones. Using this together with the fact that $\rho\in C^{0,\alpha}(K_R(a))$, we conclude that 
\begin{equation}\label{boundholdercube}
\left|\int_{K} \min_{\mathcal A_{\rho(x)}}\mathcal W dx - |K|\min_{\mathcal A_{\rho_K}}\mathcal W \right|\le C|K|r^{\alpha}\|\rho\|_{C^{0,\alpha}(K)} M_{s,d}(\rho),
\end{equation}
where
\begin{equation}\label{msdrho}
M_{s,d}(\rho) =\left\{\begin{array}{ll}\bar\rho^{s/d}\min_{\mathcal A_1}\mathcal W, &\text{ for power-law kernels,}\\ \min_{\mathcal A_1}\mathcal W - \log\underline\rho -1, &\text{ for logarithmic kernels.}\end{array}\right.
\end{equation}
We may then absorb this term into the $r^{1+\alpha}$-terms above and after summing all contributions, using the above bounds, and dividing both sides by $|K_R|$, we obtain:
\begin{align*}
\frac{\mathcal W_\eta(\tilde E,K_R(a))}{|K_R|} \le
&\,\frac{1}{|K_R|}\int_{K_R(a)}\min_{\mathcal{A}_{\rho(x)}} \mathcal W\,dx\\
&+ C (g(\eta) +1) o_{R\to \infty}(1)+o_{\eta\to0}(1)+C r^{2\alpha+2}\|\rho\|^2_{C^{0,\alpha}(K_R(a))}\\
&+C r^{\alpha+1}\|\rho\|_{C^{0,\alpha}(K_R(a))}\left(1+(g(\eta) + 1)o_{R\to \infty}(1) +o_{\eta\to0}(1)+c_{d,s}\overline \rho g(\eta)\right)^{1/2}
\end{align*}
This gives precisely \eqref{upperbound} if we choose 
\[
 0<\lambda\le\frac{\beta}{1+\alpha}.
\]
To conclude, we remark that projecting $E$ onto gradients decreases the energy. The result follows.
\end{proof}

\subsubsection{Lower Bound}

The goal of this section is to find a lower bound for $\frac{\mathcal W_\eta(\tilde E'_n, \bar K_L(a))}{|\bar K_L(a)|}$.

\begin{proposition}\label{proplb}
Let $\tilde E'_n$ be as above and let $\alpha\in]0,1]$ such that $m_V'\in C^{0,\alpha}(K_R(a))$ with $m_V'(x)\le \overline m$ and let $\beta\in]0,1+\alpha[$. Then there exists a constant $C>0$ depending only on $d,s$ such that there holds 
 \begin{align}\label{lowerbound}
 \frac{\mathcal W_\eta(\tilde E'_n, \bar K_L(a))}{|\bar K_L(a)|}\ge&\,\frac{1}{|\bar K_L|}\int_{\bar K_L(a)}\min_{\mathcal{A}_{m_V'(x)}} \mathcal W \,dx - o_{L,n\to\infty}(1)\left(1+g(\eta)\right) - \left(1+o_{L,n\to\infty}(1)\right) o_{\eta\to 0}(1) \nonumber\\
&-CL^{2\beta}\|m'_V\|^2_{C^{0,\alpha}(\bar K_L(a))}-CL^{\beta}\|m_V'\|_{C^{0,\alpha}(\bar K_L(a))}\left(C+c_{d,s}\bar m g(\eta)\right)^{1/2}.
\end{align}
\end{proposition}
\begin{proof}
Recall that thanks to the minimality of $(x_1,\ldots,x_n)$, the charges are well-separated at a distance $r_0$ which depends only on $d,s$ and $\overline m$.

The proof of the lower bound uses a partitioning argument like the one employed to prove \cite[Thm. 1.11]{rns}, however we remove the bootstrap part of the argument. Let $\bar L$ the sidelength of $\bar K_L(a)$.

First let us apply Lemma~\ref{subdivision} to partition $\bar K_{L}(a)$ into smaller hypercubes ${K_i}(a_i)$ of sidelengths $\sim r = L^{\frac{\beta}{1+\alpha}}$ and such that $\int_{K_i(a_i)}m_V'\in \mathbb N$. It follows that
\begin{equation*}
\mathcal W_\eta(\tilde E'_n, {\bar K_L(a)})= \sum_i \mathcal W_\eta(\tilde E'_n, K_i(a_i)).
\end{equation*}
The goal is to bound from below this sum.
We may assume that 
\begin{equation}\label{assummaxw}
\mathcal W_\eta(\tilde E'_n, K_i(a_i)) \le   \min_{\mathcal A_{\overline m}} \mathcal W|K_i|:=C_m|K_i|, 
\end{equation}
for otherwise, we have a lower bound $\mathcal W_\eta(\tilde E'_n, K_i(a_i))\ge C_m|K_i|$ which will suffice.

Let $K_r(b)$ be one of the above squares $K_i(a_i)$. As a consequence of the separation of charges of $\tilde E_n'$ and of \eqref{assummaxw}, by Remark \ref{sepcariche} it follows that 
\begin{equation}\label{bonborneW}
\int_{K_r(b)\times \R^k}\yg |E'_{n,\eta}|^2\le (C_m+g(\eta)C_{d,s,\overline m}) |K_r|\ ,
\end{equation}
at least for $L,r$ large enough. Therefore the hypotheses of Proposition~\ref{goodbdry} are satisfied in $K_r(b)$ and we may apply it at scale $r$ and find a good boundary cube $K_r'(b)$ which is at distance $l$ from $\partial K_r(b)$. 
By applying Proposition~\ref{screening} to the obtained good boundary we show that there exist $ \bar K_r(b)\supset K'_r(b)$, a subset $\tilde \Lambda \subset \bar K_{L}(b)$ and a vector field $\bar E'_n$ such that 
\[
\left\{
\begin{aligned}
&-\op{div}(\yg \bar E'_n)=c_{d,s}\left(\sum_{p\in\tilde \Lambda}\delta_p - m'_V(x)\delta_{\mathbb R^d}\right) & \text{ in }\bar K_{r}(b)\times\R^k\\
&\bar E'_n\cdot\vec{\nu}=0  & \text{ on }\partial\bar K_{r}(b)\times\R^k\\
&\bar E'_n\cdot\vec{\nu}=\tilde E'_n\cdot \vec{\nu}  & \text{ on }\partial K'_{r}(b)\times [-t,t]^k\\
&\bar E'_n=\tilde E'_n , \ \tilde \Lambda=\Lambda & \text{ in } K'_{r}(b)\times [-t,t]^k\\
\end{aligned}
\right.
\]
and
\begin{align}\label{compenergy0}
 \mathcal W_\eta (\tilde E'_n, K'_r(b))\ge&\, \mathcal W_\eta(\bar E'_n, \bar K_r(b))- C r^d err_{sc}(\eta,r,t,l,1,C_1,C_2).
 \end{align}
 With parameter choices like in Section~\ref{choiceparsec} with $r,t,l,1$ in the place of $L,t,l,\ve_1$ and with $C_1$ equal to the constant $C_m + g(\eta)C_{s,d,\overline m}$ from \eqref{bonborneW}, we obtain that the $err_{sc}$ error above is controlled as $(1+g(\eta))o_{r,n\to\infty}(1)$.
By using the separation of charges as in Remark \ref{sepcariche} we find that
\begin{equation*}
\mathcal W_\eta(\tilde E_n', K_r'(b))\le \mathcal W_\eta(\tilde E_n', K_r(b)) + C g(\eta) r^{d-1}l =  \mathcal W_\eta(\tilde E_n', K_r(b)) + |K_r|g(\eta)o_{r,n\to\infty}(1).
\end{equation*}
After absorbing the above error and summing over all $K_i(a_i)$, we obtain
\begin{align}\label{compenergyebar}
\mathcal W_\eta(\tilde E'_n, \bar K_L(a))\ge&\, \sum_i \min(C_m, \mathcal W_\eta(\bar E_n',\bar K_i(a_i))) - (1+g(\eta))(1+C_1)o_{r,n\to\infty}(1)|\bar K_L(a)|.
\end{align}
Now we correct the background measure of $\bar E_n'$ on each cube $K_i(a_i)$. To do this we proceed as in the proof of Proposition~\ref{propub}. We put $\tilde m= \fint_{K_r(b)}m'_V(x)\,dx$ and we define $h$ to be the solution of 
\begin{equation*}
\left\{
\begin{aligned}
	&-\op{div} (\yg \nabla h) = c_{s,d} (m'_V(x)-\tilde m)\delta_{\mathbb R^d}\quad & \text{ in}\ \bar K_r(b)\times [-r, r]^k\\
	&\partial_\nu h=0\quad & \text{ on}\ \partial \bar K_r(b)\times [-r, r]^k\\
	&\partial_\nu h=0\quad & \text{ on the rest of }\ \partial (\bar K_r(b)\times [-r, r]^k)\\
\end{aligned}
\right.
\end{equation*}
and we put $h=0$ in $\bar K_r(b)\times (\R^k\smallsetminus [-r, r]^k)$. 
Then define $\bar E$ to be $\bar E'_n + \nabla h$ in $\bar K_r(b)$. Hence
\begin{align}
\frac{\mathcal W_\eta(\bar E,\bar K_r(b))}{|\bar K_r|} \le &\, \frac{\mathcal W_\eta(\bar E'_n, \bar K_r(b))}{|\bar K_r|}+Cr^{2+2\alpha}\|m'_V\|^2_{C^{0,\alpha}(\bar K_r(b))}\nonumber\\
&+Cr^{1+\alpha}\|m'_V\|_{C^{0,\alpha}(\bar K_r(b))}\left(C+c_{d,s}\tilde m g(\eta)\right)^{1/2}.\label{compenergy1}
\end{align}
Next, we use the fact that 
$$
\frac{\mathcal W_\eta(\bar E,\bar K_r(b))}{|\bar K_r|}\ge\min_{\mathcal{A}_{\tilde m}} \mathcal W+o_{\eta\to 0}(1).
$$
Indeed, define $\tilde K_{r}(b)$ to be the hypercube obtained by taking all the iterative reflections in the $d$ directions across faces of $\bar K_{r}(b)$ (the sidelengths of $\tilde K_r(b)$ are thus $2$ twice the sidelengths of $\bar K_{r}(b)$). Since we have zero Neumann boundary condition, we can extend $\bar E$ on $\tilde K_{r}(b)\times \R^k$ by reflection across the boundary of $\tilde K_{t}(b)$. Then we periodize $\bar E$ to have a vector-field $E$  defined on $\R^d\times\R^k$. Since all the vector fields and measures are periodic, 
$$
 \min_{\mathcal{A}_{\tilde m}} \mathcal W_\eta\le \mathcal W_\eta(E)
 =\frac{2^d\mathcal W_\eta(\bar E,\bar K_r(b))}{2^d|\bar K_r|}=  \frac{\mathcal W_\eta(\bar E,\bar K_r(b))}{|\bar K_r|}.
 $$
Similarly to the process of obtaining \eqref{boundholdercube}, with the notations $M_{s,d}(m_V')$ like in \eqref{msdrho}, we have
\[
\left|\int_{\bar K_r(b)} m_V'(x)^{1+s/d}dx - |\bar K_r|\tilde m^{1+s/d}\right|\le C|\bar K_r|r^{\alpha}\|m_V'\|_{C^{0,\alpha}(K_r(b))}M_{s,d}(m_V')
\] 
and thus we absorb this into the above $r^{1+\alpha}$-term and we obtain
\begin{align*}
\mathcal W_\eta(\bar E,\bar K_r(b))\ge&\,\int_{\bar K_{r}(b)}\min_{\mathcal{A}_{m_V'(x)}} \mathcal W \,dx-|\bar K_{r}|\left(1+g(\eta)\right)(1+C_1)o_{r,n\to \infty}(1) - |\bar K_{r}|o_{\eta\to 0}(1) \nonumber\\
&-C|\bar K_{r}|r^{2+2\alpha}\|m'_V\|^2_{C^{0,\alpha}(K_{r})}-C|\bar K_{r}|r^{1+\alpha}\|m'_V\|_{C^{0,\alpha}(\bar K_{r})}\left(C+c_{d,s}\bar m g(\eta)\right)^{1/2}.
\end{align*}

Hence, summing over all cubes and using our choice of $C_m$ in \eqref{assummaxw}, using the fact that $\bar K_L(a)$ has larger measure than the union of the $K_i(a_i)$ as well as the direct bounds on $|\bar K_r|/|K_r|$ and on the integral of $m'_V$ over the difference of these cubes, we are led to 
 \begin{align}\label{lowerboundbis}
\mathcal W_\eta(\tilde E'_n, {\bar K_L(a)})\ge&\,\int_{\bar K_L(a)}\min_{\mathcal{A}_{m_V'(x)}} \mathcal W \,dx - O(l/r)|\bar K_L|  -|\bar K_L|\left(1+g(\eta)\right)o_{r,n\to \infty}(1)  \nonumber\\
&-(1+O(l/r))|\bar K_L|o_{\eta\to 0}(1)-C|\bar K_L|r^{2+2\alpha}\|m'_V\|^2_{C^{0,\alpha}(\bar K_L(a))}\nonumber\\
&-C|\bar K_L|r^{1+\alpha}\|m_V'\|_{C^{0,\alpha}(\bar K_L(a))}\left(C+c_{d,s}\bar m g(\eta)\right)^{1/2}. 
\end{align}
Note that under the choice of parameters according to Section~\ref{choiceparsec} we have $O(l/r)=o_{r,n\to\infty}(1)$. This proves \eqref{lowerbound} and thus concludes the proof.
\end{proof}

\subsubsection{End of proof of Proposition~\ref{compgbdry}}

We recall that due to \eqref{ubmaincube} and \eqref{lbmaincube}, under the choice of parameters as in Section~\ref{choiceparsec}, we have
\begin{align*}
\frac{\mathcal W_\eta(\tilde E'_n, \bar K_L(a))}{|\bar K_L(a)|}-err\le 
\frac{\mathcal W_\eta(E'_n,K'_L(a))}{|K'_L(a)|}
\le \sigma_{0}(\underline K_L(a);m'_V)+err+o_{\eta\to 0}(1),
\end{align*}
for $err:=\left(1+g(\eta)\right)(1+C_1)o_{L,n\to \infty}(1)$. Next we use Proposition~\ref{propub} in which we choose $\rho=m_V', K_R(a)=\underline K_L(a)$ and $\beta\in ]0,1[$, and Proposition~\ref{proplb}.

The final result after including the extra errors from \eqref{upperbound}, \eqref{lowerbound} and using \eqref{petitesvar} is that on the cube $K_L'(a)$ which has sidelenghts $\sim L_i\in [\ve_1^{1/d}T_i,T_i]$ we have 
\[
\frac{1}{|\bar K_L(a)|}\int_{\bar K_L(a)} \min_{\mathcal A_{m_V'(x)} }\mathcal W dx -err_1 \le \frac{\mathcal W_\eta(E'_n,K'_L(a))}{|K'_L(a)|}\le\frac{1}{|\underline K_L(a)|}\int_{\underline K_L(a)} \min_{\mathcal A_{m_V'(x)} }\mathcal W dx +err_2, 
\]
where
\begin{eqnarray*}
err_1&:=& \left(1+g(\eta)\right)(1+C_1)o_{L,n\to \infty}(1) + o_{\eta\to 0}(1) + o_{n\to\infty}(1)\left(1+g(\eta)+o_{L\to \infty}(1)\right)^{1/2}\\
err_2&:=&\left(1+g(\eta)\right)(1+C_1)o_{L,n\to\infty}(1) + \left(1+o_{L,n\to\infty}(1)\right) o_{\eta\to 0}(1) + o_{n\to\infty}(1)\left(1+g(\eta)\right)^{1/2}.
\end{eqnarray*}

Using the facts that $m_V'$ is H\"older continuous and that in Proposition~\ref{screening} we obtained $\bar K_L,\underline K_L$ are $t$-close to $K_L$, we further find up to requiring $L$ to be larger,
\begin{equation}\label{basicstepbound}
\frac{1}{|K_L(a)|}\int_{K_L(a)} \min_{\mathcal A_{m_V'(x)} }\mathcal W dx - err_1 \le \frac{\mathcal W_\eta(E'_n,K'_L(a))}{|K'_L(a)|}\le\frac{1}{| K_L(a)|}\int_{K_L(a)} \min_{\mathcal A_{m_V'(x)} }\mathcal W dx +err_2.
\end{equation}

\section{Proof of Theorem~\ref{secondmainthm}}
\subsection{Tools for estimates on crenel boundaries}
By refining the proofs of \cite[Lem. 2.3, Prop. 2.4]{ps}, we obtain an estimate generalizing the discrepancy bound of \cite[Prop. 3.1]{rns}, which relied on the ball construction \cite{ssot}, \cite{jerrard}. It seems to be a difficult open question whether a precise analogue of the ball construction can be done in dimensions $d>2$ for the not conformal energy considered here, and in particular in cases where a suited notion of degree like in \cite{jerrard} is missing. We recall that here, compared to \cite{ps}, we are in the situation of multiplicities equal to one for $p\in\Lambda$, as considered also in \cite{rns} and \cite{ssot}. This condition is a consequence of minimality as proved in \cite[Thm. 5]{ps} (Proposition~\ref{separation} here).
\begin{proposition}\label{prodecr}
Let $\Lambda, E$ be as in \eqref{eqe}.  Consider a compactly supported cutoff function $\chi_A\in C^1(\mathbb R^{d+k},[0,1])$ with $A:=\op{spt}\chi_A$. Let $0<\alpha<\eta<1$ and assume that $\mathcal W_\eta(\chi_A, E)<\infty$. For a set $S\subset \mathbb R^d$ write $S_r:=\{x:\op{dist}(x,S)\le r\}$. Then we may write
\begin{equation}
\label{subdividechangeofw}
\w(\chi_A, E) - \mathcal W_\eta(\chi_A,E) = I + II + III, 
\end{equation}
where for constants $c,C$ depending only on $s,d$ there holds
\begin{equation}
\label{i}
c\left(\min_xm(x)\right)\#\left(\Lambda \cap A_\eta\right) \le \frac{I}{\min(\eta^{d-s},\eta^d|\log\eta|)} \le C\left(\max_xm(x)\right) \#\left(\Lambda \cap A_\eta\right).\end{equation}
If for a set $S\subset \mathbb R^d$ we denote $X_r(S):=\{(p,q):\ p\neq q,\ p,q\in\Lambda\cap S,\ |p-q|<r\}$ then 
\begin{equation}
\label{ii}
0\le II \le c_{s,d}(g(\alpha) - g(\eta))\# X_{2\eta}(A_\eta).
\end{equation}
Finally, for a constant $C$ depending only on $d,s$, 
\begin{equation}
\label{iii2}
|III|\le C\#\left(\Lambda\cap \left(\op{spt}(\nabla\chi_A)\right)_\eta\right)\eta^dg^2(\alpha)\left(\|\nab_x\chi_A(0,x)\|_{L^1(\mathbb R^d)} + \int_{\mathbb R^{d+k}}\yg|\nab \chi_A||E_\alpha| \right).\end{equation}
\end{proposition}
\begin{proof}
Define 
\[
f_{\alpha,\eta}(X):= f_\alpha(X)- f_\eta(X),\quad X\in\mathbb R^{d+k}.
\]
From the definition \eqref{feta} of $f_\eta$ we find that 
\begin{equation}\label{fae}
f_{\alpha,\eta}(X)=\left\{\begin{array}{ll}
g(\eta)-g(\alpha)&\text{ for }|X|\le\alpha,\\
g(\eta)-g(X)&\text{ for }\alpha\le|X|\le \eta,\\
0&\text{ for }|X|\ge \eta,
\end{array}\right.
\end{equation}
and as a consequence of the equation \eqref{divf} that $f_\eta$ satisfies 
\begin{equation}\label{fae2}
-\op{div}(\yg \nabla f_{\alpha,\eta})= c_{d,s}\left(\delta_0^{(\eta)} - \delta_0^{(\alpha)}\right).
\end{equation}
By the definition \eqref{defeeta} of $E_\eta$ we find 
\[
E_\eta=E_\alpha+\sum_{p\in\Lambda} \nabla f_{\alpha,\eta}(X-p).
\]
Thus we may write, using the definition \eqref{Weta} of $\mathcal W_\eta$, 
\begin{eqnarray}
\mathcal W_\alpha(\chi_A,E) - \mathcal W_\eta(\chi_A,E) &=& \int_{\mathbb R^{d+k}}\yg\chi_A\left(|E_\alpha|^2 - \left|E_\alpha + \sum_{p\in\Lambda} \nabla f_{\alpha,\eta}(X-p)\right|^2\right) \nonumber\\
&&- c_{d,s}\int_{\mathbb R^{d+k}}\chi_A\sum_{p\in\Lambda}\left(g(\alpha) \delta_p^{(\alpha)} - g(\eta) \delta_p^{(\eta)}\right).\label{wetadiff}
\end{eqnarray}
The first term on the right in \eqref{wetadiff} can be expanded to give, after using an integration by parts and \eqref{fae2}, \eqref{eqclam}, the following sum:
\begin{multline}\label{rewrite1}
-c_{d,s}\sum_{p,q\in\Lambda} \int_{\mathbb R^{d+k}}\chi_A f_{\alpha,\eta}(X-p)\left(\delta_{q}^{(\eta)} - \delta_{q}^{(\alpha)}\right) + \sum_{p,q\in\Lambda} \int_{\mathbb R^{d+k}}\yg f_{\alpha,\eta}(X-p) \nabla \chi_A \cdot \nabla f_{\alpha,\eta}(X-q)\\
-2c_{d,s}\sum_{p\in\Lambda}\int_{\mathbb R^{d+k}}f_{\alpha,\eta}(X-p)\left(\sum_{q\in\Lambda}\delta_q^{(\alpha)} - m\delta_{\mathbb R^d}\right)\chi_A +2 \sum_{p\in\Lambda} \int_{\mathbb R^{d+k}} f_{\alpha,\eta}(X-p)\yg\nabla \chi_A \cdot E_\alpha.
\end{multline}
We define the following terms which summed together give \eqref{rewrite1}:
\begin{align}
I:=& 2c_{d,s}\sum_{p\in\Lambda}\int_{\mathbb R^d}f_{\alpha,\eta}(x-p)m(x)\chi_A, \label{termi}\\
II:=& -c_{d,s}\sum_{p\neq q\in\Lambda} \int_{\mathbb R^{d+k}}\chi_A f_{\alpha,\eta}(X-p)\left(\delta_{q}^{(\eta)} + \delta_{q}^{(\alpha)}\right), \label{termii}\\
III':=&-c_{d,s}\sum_{p\in\Lambda} \int_{\mathbb  R^{d+k}}\chi_A (g(\eta) - g(\alpha))\delta_p^{(\alpha)},\label{termiii'}\\
III'':=&\sum_{p,q\in\Lambda} \int_{\mathbb R^{d+k}}\yg f_{\alpha,\eta}(X-p) \nabla \chi_A \cdot \nabla f_{\alpha,\eta}(X-q) + 2 \sum_{p\in\Lambda} \int_{\mathbb R^{d+k}} f_{\alpha,\eta}(X-p)\yg\nabla \chi_A \cdot E_\alpha\label{termiii''},
\end{align}
where in the expression $III'$ we used the properties of $f_{\alpha,\eta}(X-p)$ which equals $g(\eta)-g(\alpha)$ on the support of $\delta_p^{(\alpha)}$ and is zero on the support of $\delta_p^{(\eta)}$. 
\par The term $I$ is bounded as in \eqref{i} by noticing that only the points $p\in \Lambda\cap A_\eta$ are such that $\op{spt}f_{\alpha,\eta}\cap A\neq \emptyset$ and that $\|f_{\alpha,\eta}\|_{L^1}\le C_{d,s} \min(\eta^{d-s},\eta^d|\log\eta|)$.
\par Regarding the term $II$ we note that again, each of the terms
\begin{equation}\label{termpqinii}
\int_{\mathbb R^{d+k}}f_{\alpha,\eta}(X-p)\left(\delta_q^{(\eta)}+\delta_q^{(\alpha)}\right)\chi_A
\end{equation}
corresponding to a choice $p,q\in\Lambda$ vanishes in case $(p,q)\notin X_{2\eta}(A_\eta)$ due to the support properties of $f_{\alpha,\eta}(X-p)$ and $ \delta^{(\eta)}_q$, and it has values in the interval $[g(\eta)-g(\alpha), 0]$ due to \eqref{fae} and to the fact that $\delta_q^{(\eta)},\delta_q^{(\alpha)}$ are probability measures. This gives the bounds in \eqref{ii}.
\par The term $III'$ summed to the second line in \eqref{wetadiff} give the contribution
\[
-c_{d,s}\sum_{p\in\Lambda}\int_{\mathbb  R^{d+k}}\chi_A g(\eta)\left(\delta_p^{(\alpha)} - \delta_p^{(\eta)}\right)=c_{d,s}g(\eta)\sum_{p\in\Lambda}\left(\frac{1}{|\partial B_\eta(p)|}\int_{\partial B_\eta(p)}\chi_A - \frac{1}{|\partial B_\alpha(p)|}\int_{\partial B_\alpha(p)}\chi_A\right)
\]
which can be bounded via the first term in \eqref{iii2}. The term $III''$ is also similarly bounded by \eqref{iii2}. We now define $III$ as the sum of $III', III''$ and of the second line of \eqref{wetadiff}, and we then find the bound \eqref{iii2}, concluding the proof.
\end{proof}
From the Proposition~\ref{prodecr}, by approximating the characteristic function $1_{A\times\R^k}$ of a measurable set by $C^1$-functions $\chi_A$ which appropriately avoid the set $\Lambda$, we find the following result:
\begin{corollary}\label{chircor}
\begin{itemize}
\item If $A\subset \mathbb R^d$ is a bounded Borel set such that $\op{dist}(\partial A, \Lambda)\ge \epsilon>0$ and we chose $\eta,\alpha<\epsilon$ then we may find a decomposition $\mathcal W_\alpha(A, E)-\mathcal W_\eta(A, E)= I + II$ satisfying \eqref{i}, \eqref{ii}. 
\item If in the setting of Proposition~\ref{prodecr} we further assume that $|p-q|\ge 2\eta $ for all pairs of points $p\neq q\in \Lambda\cap A$ then we find $\mathcal W_\alpha(A, E)-\mathcal W_\eta(A, E)= I$ satisfying \eqref{i}.
\end{itemize} 
\end{corollary}
\subsection{Construction of crenel boundaries}
The main idea of the proof of Theorem~\ref{secondmainthm} is that if we perturb the boundary of $K_\ell(a)$ such that it avoids the charges then the sharp bounds of Proposition~\ref{prodecr} allow to take the $\eta\to0$ limit in Theorem~\ref{mainthm} without uncontrolled error terms. For the perturbation of $\partial K_\ell(a)$ we will need the following tool:
\begin{proposition}[crenel boundaries]\label{creneau}
Let $m_V,m_V', \Sigma,\Sigma', E_n'$ be as in Theorem~\ref{secondmainthm} and $K_\ell(a)\subset \Sigma'$. Let moreover $r_0$ be the minimum point separation of the charges corresponding to $E_n'$, bounded in Proposition \ref{separation}. Then the following hold.
\begin{enumerate}
\item There exists a set $\Gamma$ which can be expressed as the image of a bi-Lipschitz deformation $f:K_\ell(a)\to \Gamma$ such that $\|f - id\|_{L^\infty}\le 1$, and for which 
\begin{equation}\label{creneauevita}
\min_{p\in\Lambda\cap \Gamma}\op{dist}(p,\partial \Gamma)\ge \frac{r_0}{8}.
\end{equation}
\item There exists a constant $C$ depending only on the dimension such that if
\begin{equation}\label{conditioneta}
r_1<Cr_0^dL^{-d+1}
\end{equation}
then there exists a universal constant $C$ and a cube $K_\ell'(a)$ with $\op{dist}(\partial K_\ell'(a),\partial K_\ell(a))<1$ and such that moreover 
\begin{equation}\label{klevita}
\min_{p\in\Lambda\cap K_\ell'}\op{dist}(p,\partial K'_\ell)\ge r_1.
\end{equation}
\end{enumerate}
\end{proposition}
\begin{proof}
\textbf{Part 1.} This follows Step 1 of the proof of \cite[Prop. 5.6]{rs}. Consider the set $P$ points $p$ such that $K_{r_0/2}(p)\cap \partial K_\ell\neq\emptyset$. Then all cubes $K_{r_0/2}(p),\ p\in P$ are disjoint and included in the set $K_{\ell+1}(a)\setminus K_{\ell-1}(a)$. 
Then define $\Gamma:=K_\ell\cup\bigcup_{p\in P} K_{r_0/2}(p)$. The fact that $\Gamma$ is a bi-lipschitz and close to the identity deformation of $K_\ell$ is a straightforward but tedious argument that we leave to the reader.\\

\textbf{Part 2.} We prove that if \eqref{conditioneta} holds for a small enough $C$ depending on $d$ only, then $\Gamma$ can be chosen to be a hyperrectangle. Note first that due to the separation condition, if $r_0<1$ then the $r_0/2$-balls with centers in $P$ are disjoint and contained in $K_{\ell+1}\setminus K_{\ell-1}$, therefore by a straightforward volume comparison the number of charges in $P$ satisfies
\begin{equation}\label{boundnumcharg}
|P|\le \frac{|K_{\ell+1}\setminus K_{\ell-1}|}{|B_{r_0/2}|} \le C r_0^{-d} \ell^{d-1},
\end{equation}
where $C_{\eqref{boundnumcharg}}$ is a packing constant which depends only on the dimension. Let now
\[
T_{r_1}:=\left\{\tau\in[-1,1]:\ \exists p\in \Lambda, \partial K_{\ell+\tau} \cap B_{r_1}(p)\neq \emptyset\right\},
\]
Then $T_{r_1}$ can be covered by at most $|P|$ intervals of size $2r_1$ and thus due to the bound \eqref{boundnumcharg}, we find that $|T_{r_1}|\le C_{\eqref{boundnumcharg}}r_1 r_0^{-d}\ell^{d-1}$. Now fix $C_{\eqref{conditioneta}}$ depending on the packing constant $C_{\eqref{boundnumcharg}}$ only, such that for $r_1$ satisfying \eqref{conditioneta} we have $C_{\eqref{boundnumcharg}}r_1 r_0^{-d}L^{d-1}<1$. Therefore $[-1,1]\setminus T_{r_1}\neq \emptyset$, furnishing the desired cube $K_\ell'(a)$.
\end{proof}

\subsection{Proof of Theorem~\ref{secondmainthm} given Theorem~\ref{mainthm}}
For each $K_\ell(a_n)$ as in Theorem~\ref{secondmainthm} we consider a modification $\Gamma_n$ as obtained by applying Proposition~\ref{creneau} (the case of cubes $K_\ell'(a_n)$ is completely analogous). Then firstly, we obtain via Proposition~\ref{prodecr} and Corollary \ref{chircor} that for all $0<\alpha<\eta<r_0/2$, where $r_0$ is the separation constant of Proposition~\ref{separation}, there holds
\begin{equation}\label{erroreta}
\left|\mathcal W_\eta(E_n',\Gamma_n) - \mathcal W_\alpha(E_n',\Gamma_n)\right|\le C\bar m \min(\eta^{d-s},\eta^d|\log\eta|) |\Gamma_n|,
\end{equation}
Second, recall that by the result \eqref{estimationproofth} together with the generalization of such bounds, which is done as in Section~\ref{proofnocrenel}, we have that for the choices $\sigma=\pm1$ there holds
\begin{equation}\label{boundsoncube}
\left|\frac{\mathcal W_\eta(E'_n,K_{\ell + \sigma}(a_n))}{|K_{\ell + \sigma}(a_n)|}-\frac{1}{| K_{\ell + \sigma}|}\int_{K_{\ell + \sigma}(a_n)} \min_{\mathcal A_{m_V'(x)} }\mathcal W dx\right|\le (1+g(\eta))(1+C_1)o_{\ell,n\to\infty}(1)+o_{\eta\to0}(1).
\end{equation}
Moreover, if $M_{s,d}$ is the shorthand used in the bound \eqref{msdrho}, we obtain
\begin{equation}\label{boundeasygamman}
\left|\int_{K_{\ell+1}(a_n)\setminus\Gamma_n}\min_{\mathcal A_{m_V'(x)}}\mathcal W dx\right|\le M_{s,d}(m_V)|K_{\ell+1}(a_n)\setminus\Gamma_n|\le C M_{s,d}\ell^{d-1}.
\end{equation}
Using the bounds \eqref{boundsoncube}, the fact that $K_{\ell-1}(A_n)\subset \Gamma_n\subset K_{\ell+1}(a_n)$, the definition of $\mathcal W_\eta$ (and in particular the fact that the two terms defining $\mathcal W_\eta(E_n',A)$ as in \eqref{Weta} are additive and monotone under inclusion with respect to $A$) and Remark \ref{sepcariche}, we obtain
\begin{equation}\label{changetogamma}
\left|\frac{\mathcal W_\eta(E_n',K_{\ell+1}\setminus\Gamma_n)}{|\Gamma_n|}\right|\le (1+g(\eta))(1+C_1)o_{\ell,n\to\infty}(1)+o_{\eta\to0}(1).
\end{equation}

By summing up \eqref{erroreta}, \eqref{boundsoncube}, \eqref{boundeasygamman} and \eqref{changetogamma}, we find that for all $0<\alpha<\eta<r_0/2$ there holds
\begin{equation}\label{secondmainthmfinal}
\left|\frac{\mathcal W_\alpha(E'_n,\Gamma_n)}{|\Gamma_n|}-\frac{1}{| \Gamma_n|}\int_{\Gamma_n} \min_{\mathcal A_{m_V'(x)} }\mathcal W dx\right|\le (1+g(\eta))(1+C_1)o_{\ell,n\to\infty}(1)+o_{\eta\to0}(1).
\end{equation}
Therefore for all $\epsilon>0$ we may fix $\eta<r_0/2$ such that the rightmost term above is $\le \epsilon$ and then let $\alpha\to 0$ obtaining that the limit in \eqref{mainlimitbis} is $\le\epsilon$. As $\epsilon>0$ is arbitrary this concludes the proof of \eqref{mainlimitbis}.
\section{Discrepancy bounds}
In this section we show how to deduce the result of Theorem~\ref{discrepancybound} from the one of Theorem~\ref{mainthm}. We consider a scale $\underline \ell\le\tfrac{1}{4}\ell$ to be more precisely fixed later, and we find that due to Theorem~\ref{mainthm} there holds 
\begin{equation*}
\mathcal W_\eta(E_n',K_{\underline\ell}(b)) \le \int_{K_{\underline \ell}(b)}\min_{\mathcal A_{m'_V(x)}}\mathcal W dx + |K_{\underline\ell}|o_{\eta\to0, n,\underline \ell\to\infty}(1).
\end{equation*}
Due to the charge separation condition and to Remark \ref{sepcariche} and to the bounds on $m_V'$, we find, using the scaling of $c\mapsto\min_{\mathcal A_c}\mathcal W$ (see \eqref{msdrho}) that 
\begin{equation}\label{c1bar}
\int_{K_{\underline \ell}(b)}|E_n'|^2 \le \int_{K_{\underline \ell}(b)}\min_{\mathcal A_{m'_V(x)}}\mathcal W dx + C_{d,\overline m}(1+g(\eta))|K_{\underline\ell}|\le \bar C_1|K_{\underline \ell}(b)|,
\end{equation}
where $\bar C_1$ depends only on $\underline m, \overline m, s, d$. By Besicovitch's covering theorem, we find a cover of $\partial K_{\ell - 2\underline \ell}(a)$ by the union of $J_d$ families $\mathcal F_j$ of cubes $K_{\underline \ell}(b),b\in\partial K_{\ell - 2\underline\ell}$, such that the cubes $K_{2\underline \ell}(b)$ belonging to a given $\mathcal F_j$ are disjoint. Then for each $j\in J_d$ and each $K\in\mathcal F_j$, from \eqref{c1bar} by a mean value theorem there exist $\underline\ell_1\in[\underline \ell,2\underline\ell]$ (note that $\underline\ell_1$, unlike $\underline\ell$, depends on the center $b$, but we omit this dependence to make notations lighter) such that 
\begin{equation}\label{goodslice}
\int_{\partial K_{\underline\ell_1}(b)\times\mathbb R^k} |E_n'|^2\le C\bar C_1\underline\ell^{d-1}.
\end{equation}
Then from \eqref{goodslice} we obtain by Cauchy-Schwartz inequality
\begin{equation}\label{bek=0}
\left|\int_{\partial K_{\underline \ell_1}(b)} E_n'\cdot \nu\right|\le|\partial K_{\underline\ell_1}(b)|^{1/2}\left(\int_{\partial K_{\underline\ell_1}(b)}|E_n'|^2\right)^{1/2}\le C \bar C_1\underline \ell^{d-1}.
\end{equation}
As a consequence of the properties of our cover, we also obtain that the cubes $K_{\underline\ell_1}(b)$ cover $\partial K_{\ell-2\underline \ell}(a)$. In particular, 
\[
\partial R_{\ell,\underline\ell}:=\partial\left(K_{\ell-2\underline \ell}(a)\cup\bigcup_{j\in J_d}\bigcup_{K_{\underline \ell}\in\mathcal F_j}K_{\underline\ell_1}\right) \subset \partial\left(\bigcup_{j\in J_d}\bigcup_{K_{\underline \ell}\in\mathcal F_j}K_{\underline\ell_1}\right).
\]
It then follows that, using \eqref{bek=0} and the fact that the cubes $K_{\underline\ell_1}$ corresponding to a single $\mathcal F_j$ are disjoint, 
\begin{equation}\label{unionbek=0}
\left|\int_{\partial R_{\ell,\underline\ell}}E_n'\cdot \nu\right|\le\sum_{j\in J_d}\sum_{K_{\underline\ell}\in\mathcal F_j}\left|\int_{\partial K_{\underline\ell_1}}E_n'\cdot\nu\right|\le|J_d|\max_{j\in J_d}|\mathcal F_j| C \bar C_1\underline\ell^{d-1}.
\end{equation}
As the $K_{\underline \ell}$ cover $\partial K_{\ell-2\underline \ell}(a)$, also for the larger cubes $K_{\underline\ell_1}$ the sets $\partial K_{\ell-2\underline\ell}(a)\cap K_{\underline\ell_1}$ corresponding to the union of all the $\mathcal F_j$'s cover $\partial K_{\ell - 2\underline\ell}(a)$.
We also have that, as the centers of such $K_{2\underline\ell}$ belong to $\partial K_{\ell-2\underline\ell}(a)$, there holds
\[
\mathcal H^{d-1}(\partial K_{\ell-2\underline \ell}(a)\cap K_{2\underline\ell})\le C_d \underline\ell^{d-1}. 
\]
Therefore, as the $K_{2\underline \ell}(b)$ corresponding to each single family $\mathcal F_j$ are disjoint, by summing the above inequalities over each fixed $\mathcal F_j$ separately, we find the lower bound below, while the upper bound is straightforward:
\[
\max_{j\in J_d}|\mathcal F_j|C_d \underline\ell^{d-1}\le\mathcal H^{d-1}(\partial K_{\ell-2\underline\ell}(a))\le C(\ell-2\underline\ell)^{d-1}.
\]
Using this and \eqref{unionbek=0}, by integrating the equation \eqref{hnp} satisfied by $E_n'$ over $R_{\ell,\underline\ell}$ and using Stokes' theorem, we find that (for a new constant $C$ depending on the above $C_d$ and using the previous choice $\underline\ell<\tfrac14\ell$ as above)
\begin{equation}\label{almostdiscrep}
\left|\nu_n'(R_{\ell,\underline\ell}) -\int_{R_{\ell,\underline\ell}}\mu_V'\right|=\left|\int_{\partial R_{\ell,\underline\ell}}E_n'\cdot \nu\right|\le C \bar C_1 |J_d| (\ell-2\underline\ell)^{d-1} \le C \bar C_1(\ell-2\underline\ell)^{d-1}\le 2^{1-d}\bar C_1\ell^{d-1}.
\end{equation}
To reach the desired bound, as \eqref{almostdiscrep}, \eqref{discrepdecay} involve additive quantities, by triangle inequality we just need to estimate $|\nu_n' - \int\mu_V'|$ over $K_{\ell}(a)\setminus R_{\ell,\underline\ell}$. To this aim we separately bound from above $\nu_n'(K_{\ell}(a)\setminus R_{\ell,\underline\ell})$ and $\int_{K_{\ell}(a)\setminus R_{\ell,\underline\ell}}\mu_V'$ by using Remark \ref{sepcariche} and the bounds on $m_V$, together with the fact that $|K_{\ell}(a)\setminus R_{\ell,\underline\ell}|\le C \underline\ell\ell^{d-1}$. Therefore we obtain 
\begin{equation}\label{discrep2}
\left|\nu_n'(K_\ell(a)) -\int_{K_\ell(a)}\mu_V'\right|\le C (\bar C_1 + (1+g(\eta))\underline\ell)\ell^{d-1}.
\end{equation}
Here $\bar C_1$ depends only on the bounds in \eqref{mainlimit} at scale $\underline\ell$, thus $\underline\ell$ can be chosen so that the energy error in \eqref{mainlimit} at scale $\underline\ell$ is uniformly bounded on all $\underline\ell$-cubes. As the right hand side of \eqref{discrep2} does not have any further dependece on $\underline\ell$, we find \eqref{discrepdecay}.

%\bibliographystyle{siam}
%\bibliography{bibequidistr}

\end{document}